\documentclass[article,onecolumn, draftclsnofoot]{IEEEtran}  

\usepackage[utf8]{inputenc}
\usepackage[T1]{fontenc}
\usepackage[hidelinks]{hyperref}
\usepackage{url}
\usepackage{ifthen}
\usepackage{cite}
\usepackage[cmex10]{amsmath} 
\usepackage{amsfonts,amssymb}
\usepackage{color}
\usepackage{ifthen}
\usepackage[arxiv]{optional} 
\usepackage{bbm}
\usepackage{cite}
\usepackage{amsmath,amssymb,amsfonts}
\usepackage{algorithmic}
\usepackage{graphicx}
\usepackage{textcomp}
\usepackage{xcolor}
\usepackage{caption}
\usepackage{subcaption}
\usepackage{graphics} 
\usepackage{epsfig} 
\usepackage{times} 
\usepackage{amsmath} 
\usepackage{amssymb}  
\usepackage{caption}
\usepackage{hyperref}
\usepackage{bbold}

\newcommand{\eq}[1]{\begin{align}#1\end{align}}
\newcommand{\seq}[1]{\begin{subequations}#1\end{subequations}}

\newcommand{\lb}[1]{\Bigg\{ \begin{array}{ll} #1 \end{array} }

\newcommand{\bit}[1]{\begin{itemize}#1\end{itemize}}
\newcommand{\E}{\mathbb{E}}
\newcommand{\p}{\mathbb{P}}

\newcommand{\cT}{[T]}

\newcommand{\cP}{\mathcal{P}}

\newcommand{\cX}{\mathcal{X}}

\newcommand{\cA}{\mathcal{A}}

\newcommand{\cZ}{\mathcal{Z}}

\newcommand{\tsigma}{\tilde{\sigma}}
\newcommand{\tgamma}{\tilde{\gamma}}

\newcommand{\defeq}{\buildrel\triangle\over =}

\newcommand{\nn}{\nonumber}

\newcommand{\cS}{\mathcal{S}}

\DeclareMathAlphabet{\mathcal}{OMS}{cmsy}{m}{n}

\usepackage{lipsum}
\usepackage{amsfonts}
\usepackage{graphicx}
\usepackage{epstopdf}
\usepackage{algorithmic}
\ifpdf
  \DeclareGraphicsExtensions{.eps,.pdf,.png,.jpg}
\else
  \DeclareGraphicsExtensions{.eps}
\fi

\newtheorem{lemma}{Lemma}

\newtheorem{theorem}{Theorem}

\newtheorem{proof}{Proof}
\newtheorem{definition}{Definition}
\newtheorem{claim}{Claim}

\title{Master Equation for Discrete-Time Stackelberg Mean Field Games with single leader}

\author{Deepanshu Vasal and Randall Berry\thanks{The authors are with the Department of Electrical and Computer Engineering, Northwestern University,
  ({dvasal@umich.edu}, {rberry@northwestern.edu}).
  }
  }

\usepackage{amsopn}

\ifpdf
\fi

\begin{document}

\maketitle

\begin{abstract}
  In this paper, we consider a discrete-time Stackelberg mean field game with a leader and an infinite number of followers. The leader and the followers each observe types privately that evolve as conditionally independent controlled Markov processes. The leader commits to a dynamic policy and the followers best respond to that policy and each other. Knowing that the followers would play a mean field game based on her policy, the leader chooses a policy that maximizes her reward. We refer to the resulting outcome as a Stackelberg mean field equilibrium (SMFE). In this paper, we provide a master equation of this game that allows one to compute all SMFE. Based on our framework, we consider two numerical examples. First, we consider an epidemic model where the followers get infected based on the mean field population. The leader chooses subsidies for a vaccine to maximize social welfare and minimize vaccination costs. In the second example, we consider a technology adoption game where the followers decide to adopt a technology or a product and the leader decides the cost of one product that maximizes his returns, which are proportional to the people adopting that technology.  
    
\end{abstract}

\section{Introduction}

%
%

With increasing amount of integration of technology in our society and with recent advancements in computation and algorithmic technologies, there is an unprecedented scale of interaction among people and devices. As an example, smartphones have penetrated our society in the last decade and more than 65\% of the world is connected through the internet.  
Such deep inter-connectedness of the society demands a  need to design and understand the behavior of the resulting \emph{large scale} interactions and a need to design policies by the government and  private players to better govern and optimally respond. In this paper, we present a new methodology to analyze such interactions through Stackelberg mean-field dynamic games.

The theory of dynamic games is a powerful tool to model such sequential strategic interaction among selfish players, introduced by \cite{Sh53}. Discrete-time dynamic games with Markovian structure have been studied extensively to model many practical applications in both engineering and economics, such as as dynamic auctions~\cite{IyJoSu14,BeSa10}, security~\cite{EtBa19}, markets~\cite{Vill99,BoPaWu18}, traffic routing~\cite{JoScSt04,MePaOzAc20}, wireless systems~\cite{AdJoGo07}, social learning~\cite{VaAn16allerton_extended, LeSuBe14}, oligopolies-- i.e. competition among firms (e.g.~\cite{BaOl98, FiVr12}), and more. 

In dynamic games with perfect and symmetric information, subgame perfect equilibrium (SPE) is an appropriate equilibrium concept.
Markov Perfect Equilibria (MPE), introduced in \cite{MaTi01}, is a refinement of SPE that is also used, where players' strategies depend on a coarser Markovian state of the systems, instead of the whole history of the game which grows exponentially with time and thus becomes unwieldy.  An analogous notion to SPE for incomplete information games is perfect Bayesian equilibrium (PBE).
However, when the number of players is large, computing MPE/PBE becomes intractable.
To model the behavior of large population strategic interactions, mean-field games were introduced independently by \cite{HuMaCa06}, and \cite{LaLi07}. In such games, there are large number of homogeneous strategic players, where each player has infinitesimal effect on system dynamics and is affected by other players through a mean-field population state. There have been a number of applications such as economic growth, security in networks, oil production, volatility formation, population dynamics (see ~\cite{La08,GuLaLi11,SuMa19,HuMa16,HUMa17,HuMa17cdc, AdJoWe15} and references therein). 

An engineering side of game theory is the theory of Mechanism design that deals with the \emph{design} of games such that when acted upon by the strategic players, the equilibrium(s) of the game coincide with the outcome desired by the designer. It could be social welfare say desired by the government or profit maximization desired by private entities. 
Stackelberg equilibrium (SE) is a notion of equilibrium related to mechanism design.
A Bayesian Stackelberg game is played between two players: a leader and a follower.
The follower has a private type that only she observes, however, the leader knows the prior distribution
on that state. The leader commits to a strategy that is observable to the follower. The
follower then plays a best response to leader’s strategy to maximize its utility. Knowing
that the follower will play a best response, the leader commits to and plays a strategy
that maximizes his utility. Such pair of strategies of the leader and the follower is called a Stackelberg equilibrium. It is
known that such strategies can provide higher utility to the leader than that obtained in a Nash
equilibrium of the game.

In this paper, we consider discrete-time Stackelberg mean field games where there is a leader and infinitely many followers. The leader and each follower sequentially make strategic decisions and are affected by other players through a mean-field population state of the followers' and the leader's actions. 
Each follower has a private type that evolves through a controlled Markov process which only she observes and leader and all the followers observe the current population state which is the distribution of all the followers' types. As before, the leader commits to a policy and all the followers best respond to it while being in equilibrium with each other such that for each time $t$ and given the leader's policy, a follower's policy (symmetric across all followers) maximizes her reward to go so that she doesn't gain by unilaterally deviating while all the other followers play the equilibrium policy. Similarly the leader plays a strategy such that when all the followers' best respond to the leader's strategy and are in equilibrium with each other, then the leader's policy maximizes her reward to go. A special case is when the leader is social welfare maximizing and her instantaneous reward is sum of expected reward of the followers, where expectation is defined through the mean field state.

In such games, a Stackelberg Mean Field Equilibrium (SMFE) is defined through a coupled fixed-point equation as follows: the mean-field state evolves through a Fokker-Planck \emph{forward} equation given an SMFE policy profile of the leader and the  followers. The followers' (symmetric) policy is a best response to the leader's equilibrium policy, given the mean-field state evolution process. Finally, the leader's policy is optimum given that followers play the best response. 
As a result, in order to compute an SMFE, one needs to solve a coupled fixed-point equation in the space of mean-field states and the equilibrium policies of the leader and the followers.
In principle, one can solve this fixed-point equation across time for the whole game; however, the resulting complexity will increase double exponentially with time. In this paper, we present an algorithm which can equivalently solve for smaller fixed-point equations for each time $t$, thereby reducing the complexity to linear in time.
This algorithm is equivalent to the {\it master equation} of continuous-time mean field games~\cite{CaDeLaLi15} that allows one to compute all mean field equilibria (MFE) of the game sequentially.
  
     Our algorithm is motivated by the developments in the theory of dynamic games with asymmetric information in~\cite{VaSiAn16arxiv,VaAn16allerton, VaAn16cdc, Va20acc ,VaMiVi21, Va20St}, where authors in these works have considered different models of such games and provided a sequential decomposition framework to compute Markovian perfect Bayesian equilibria and Stackelberg equilibria of such games. 

Using our framework, we consider two problems. The first example we consider is a malware spread problem in a cyber-physical system where the followers correspond to nodes in the system. Nodes get infected by an independent random process where each node has a higher risk of getting infected if there are more infected nodes in the system, due to negative externality imposed by other infected nodes. At each time $t$, each follower privately observes her own state and publicly observes the population of infected nodes, based on which she has to make a decision to repair or not.
Furthermore, there is a leader (e.g.~a government) that decides on subsidies to take the action `repair' with the goal of maximizing social welfare and minimizing the total cost of the subsidy. In the second example we consider technology adoption where we consider two technologies, one with a constant exogenous price and second with a price set by the firm, who dynamically prices its product to get most customers. We assume there is some ``stickyness" to the product such that if a buyer chooses that product, it is more likely to prefer that next time as well. 

The paper is structured as follows. In Section~\ref{sec:Model}, we present the model, our notation and background. In Section~\ref{sec:Prelim}, we present the notion of a Stackelberg Mean Field Equilibrium (SMFE) and the common information approach.
In Section~\ref{sec:Result}, we present our main results, where we present an algorithm to compute a SMFE for the finite horizon game. We present numerical examples in Section~\ref{sec:Example}. We conclude in Section~\ref{sec:Conclusion}. All proofs are presented in Appendix.

\subsection{Notation}
We use uppercase letters for random variables and lowercase for their realizations. For any variable, subscripts represent time indices and superscripts represent player identities. We use notation $ -i$ to represent all players other than player $i$ i.e. $ -i = \{1,2, \ldots i-1, i+1, \ldots, N \}$. We use notation $a_{t:t'}$ to represent the vector $(a_t, a_{t+1}, \ldots a_{t'})$ when $t'\geq t$ or an empty vector if $t'< t$. We use $a_t^{-i}$ to mean $(a^1_t, a^2_{t}, \ldots, a_t^{i-1}, a_t^{i+1} \ldots, a^N_{t})$ . We remove superscripts or subscripts if we want to represent the whole vector, for example $a_t$  represents $(a_t^1, \ldots, a_t^N) $. We denote the indicator function of any set $A$ by $\mathbbm{1}\{A\}$. 
For any finite set $\mathcal{S}$, $\mathcal{P}(\mathcal{S})$ represents space of probability measures on $\mathcal{S}$ and $|\mathcal{S}|$ represents its cardinality. Given a set $\mathcal A$, we denote its $n$-fold Cartesian product by $({\mathcal A})^n$.   We denote the set of real numbers by $\mathbb{R}$. For a probabilistic strategy profile of players $(\sigma_t^i)_{i\in [N]}$ where probability of action $a_t^i$ conditioned on $z_{1:t},x_{1:t}^i$ is given by $\sigma_t^i(a_t^i|z_{1:t},x_{1:t}^i)$, we use the short hand notation $\sigma_t^{-i}(a_t^{-i}|z_{1:t},x_{1:t}^{-i})$ to represent $\prod_{j\neq i} \sigma_t^j(a_t^j|z_{1:t},x_{1:t}^j)$.   We denote by $P^{\sigma}$ (or $E^{\sigma}$) the probability measure generated by (or expectation with respect to) strategy profile $\sigma$.
All equalities and inequalities involving random variables are to be interpreted in the \emph{a.s.} sense. For any variable $a$, we define $\cS_a$ as the space of all possible $a$.

\section{Model}
\label{sec:Model}
We consider a stochastic Stackelberg mean field game over a time horizon $[T]\defeq$ $\{1, 2, \ldots T\}$ with perfect recall as follows. Suppose there are two kinds of players: a leader and an infinite number of followers. Both the leader and the followers have private types, $x_t^l \in \cX^l$ for the leader and $x_t^{f,i} \in \cX^f $, for the follower $i$ at time $t$, where $x_t^{f,i},x_t^l$ evolve as a conditionally independent controlled Markov processes in the following way, where for any finite $N$ number for followers,

\eq{
P(x_t^l,x_t^{f,1},\ldots, x_t^{f,N}|z_{1:t-1},a_{1:t-1},x_{1:t-1}) &=  Q(x_t^l|z_{t-1},a_{t-1},x_{t-1}^l )\prod_{i=1}^NQ(x_t^{f,i}|z_{t-1},a_{t-1},x_{t-1}^l,x_{t-1}^f ),
}
where  $a_t = (a_t^l,a_t^f)$ is the actions taken by the leader and followers at time $t$ and $Q$ is a known kernel. The leader takes action $a_t^l\in \cA^l$ at time $t$ on observing $z_{1:t},x_{1:t}^l$, and the follower $i$ takes action $a_t^{f,i}\in \cA^f$ at time $t$ on observing $z_{1:t}$ and $x_{1:t}^{f,i}$,
where $z_t$ is the mean field population state of the followers at time $t$, i.e.,
\eq{
z_t(x) \defeq \lim_{N\to\infty} \sum_{i=1}^N \frac{1}{N}1(x_t^{f,i} = x).
}
Here, $z_{1:t}$ is common information among players, and $x_{1:t}^l(x_{1:t}^f)$ is private information of the leader (and the followers, respectively). We denote the set of possible values of the mean-field state by $\cZ$.

 At the end of interval $t$, the leader receives an instantaneous reward $R_t^l(z_t,x_t^l,a_t^l)$ and the follower $i$ receives an instantaneous reward $R_t^f(x_t^l,x_t^{f,i},a_t^{f,i},a_t^l,z_t)$. Note that the leader's reward only depends on the followers' actions through the mean-field state. Likewise for each follower, their reward depends on the actions of the other followers through the mean-field state, but does depend directly on the action of the leader and the follower's own action.

The sets $\cA^l,\cA^f, \cX^l,\cX^f $ are assumed to be finite. Let $\sigma^i = ( \sigma^i_t)_{t \in [T]}$ be a probabilistic strategy of player $i\in\{l,f \}$ where $\sigma^l_t : (\mathcal{Z})^{t}\times(\cA^l)^{t-1}\times(\cX^l)^t \to \mathcal{P}(\cA^l)$, and $\sigma^f_t : (\mathcal{Z})^{t}\times (\cA^l)^{t-1}\times(\cX^f)^t \to \mathcal{P}(\cA^f)$ such that the leader plays action $A_t^l$ according to $ A_t^l \sim \sigma^l_t(\cdot|z_{1:t},a_{1:t-1}^l,x_{1:t}^l)$, and the follower plays action $A_t^f$ according to $ A_t^f \sim \sigma^f_t(\cdot|z_{1:t},a_{1:t-1}^l,x_{1:t}^f)$, Let $ \sigma \defeq(\sigma^i)_{i\in \{l,f\}}$ be a strategy profile of all players. Suppose players discount their rewards by a discount factor $\delta\leq 1$.

\section{Preliminaries}
\label{sec:Prelim}
In this section, we first present the definition of a Stackelberg Mean field equilibrium (SMFE) which we will use in this paper. We then discuss the common agent approach that we will utilize in deriving an algorithm for finding an SMFE.

\subsection{Stackelberg mean field equilibrium}
\label{sec:PBSE}
In this paper, we will consider followers' Markovian equilibrium policies that only depend on their current states $x_t^f$, current mean field state $z_t$ and a common belief $\pi_t$, where $\pi_t(x_t^l) = P^{\sigma^f,\sigma^l}(x_t^l|z_{1:t},a_{1:t-1}^l)$ i.e. $\pi_t$ is the common belief on the leader's state given the common information $(z_{1:t},a_{1:t-1}^l)$. Thus, at equilibrium $a_t^{f,i}\sim \tsigma^{f,i}_t(\cdot|\pi_t,z_t,x_t^{f,i})$ and the leader's strategy $a_t^{l}\sim \tsigma^{l}_t(\cdot|\pi_t,z_t,x_t^l)$.\footnote{Note, however, that for the purpose of equilibrium, we allow for deviations in the space of all possible strategies that may depend on the entire observation history.} 

For the game considered, we first define several mappings as follows that we will use in turn to define a Stackelberg mean field equilibrium.


%
Let $BR_t^f: \cP(\cX^l)\times\mathcal{Z}^t\times (\cA^l)^{t-1}\times (\mathcal{X}^f)^{t}\times (\mathcal{\sigma}^l)^{T-t}\to \mathcal{S}_{\sigma^f}$ be defined by
\eq{\label{eq:BR1}
BR_t^f(\pi_t,z_{1:t},a_{1:t-1}^l,x_{1:t}^{f,i},\sigma_{t:T}^l) &:=\arg\max_{\sigma^f} \E^{\sigma_{t:T}^l,{\sigma}_{t:T}^f,\pi_t}[\sum_{n=t}^T \delta^{n-t} R^f(X_n^l,X_n^{f,i},A_n^{f,i},Z_n)|\pi_t,z_{1:t},a_{1:t-1}^l,x_{1:t}^{f,i} ].
}
 This specifies a follower's best response at time $t$ given the history of the mean-field state and its private type up to time $t$ and the leader's strategy from time $t$ on-wards.   
 Note that this mapping specifies a complete policy for the follower for all time $t$. 
 Next, let  $BR^f: \cZ^T\times(\cA^l)^{T-1} \cS_{\sigma^l}\to{\cS_{\sigma^f}}$ be given by
\eq{
BR^f(z_{1:T},\sigma^l) &:=\bigcap_{t}\bigcap_{a_{1:t-1}^l}\bigcap_{x_{1:t}^{f,i}}BR_t^f(\pi_t,z_{1:t},a_{1:t-1}^l,x_{1:t}^{f,i},\sigma_{t:T}^l).
}
This specifies the follower's best response policies which are consistent with (\ref{eq:BR1}) for all times $t$ as a function of the mean-field state trajectory $z$ and the leaders policy $\sigma^l$. 

%

With some abuse of notation, we will also say $\sigma_t^f\in BR_t^f(z_{1:t},x_{1:t}^{f,i},\sigma_{t:T}^l)$ if there exists $\hat{\sigma}^f \in BR_t^f(z_{1:t},x_{1:t}^{f,i},\sigma_{t:T}^l)$ such that $\sigma_t^f = \hat{\sigma}_t^f$.

Conversely, define a mapping $\Lambda: \cS_{\sigma^f}\times \cS_{\sigma^l} \to \cZ^T $ as follows: given $\sigma^f \in S_{\sigma^f},\sigma^l \in S_{\sigma^l}, z = \Lambda(\sigma^f,\sigma^l)$, is constructed recursively as
$\forall t,z_{1:t},a_{1:t-1}^l,x_{1:t}^f, x_{1:t}^l$ 
\eq{z_{t+1}(\cdot) = \sum_{x_t^f,a_t}z_t(x_t^f)\pi_t(x_t^l) Q( \cdot|x_t^f, a_t^f,a_t^l,z_t)\sigma^f_t(a^f_t|z_{1:t},a_{1:t-1}^l,x_{1:t}^f)\sigma_t^l(a_t^l|z_{1:t},a_{1:t-1}^l,x_{1:t}^l). 
}
This mapping determines the mean-field trajectory as a function of the leader's and the follower's policies. Finally, let
\eq{ BR^l(z) &:=\bigcap_t \bigcap_{x_{1:t}^l}  \arg\max_{ \sigma^l} \E^{\sigma^l,\hat{\sigma}^f,\pi_t} \big\{ \sum_{n=t}^T \delta^{n-t}R_n^l(X_n^l,Z_n,A_n) |\pi_t,z_{1:t},a_{1:t-1}^l,x_{1:t}^l\big\},\\
&\text{ where, } \hat{\sigma}^f \in BR^f(z,\sigma^l).
}



\begin{definition}
\label{def:MFE}
A tuple $(\tsigma^f,\tsigma^l,z)$ is a Stackelberg mean-field equilibrium (SMFE) if

[(a)]: $\tsigma^f\in BR^f(z,\tsigma^l)$,

[(b)]: $z=\Lambda(\tsigma^l,\tsigma^f)$, and

[(c)]: $\tsigma^l \in BR^l(z)$

\end{definition}



\subsection{Common agent approach}
We recall that in general, the leader and the followers generate their actions at time $t$ as follows, $a_t^l\sim \sigma_t^l(\cdot|z_{1:t},a_{1:t-1}^l,x_{1:t}^l)$ and $a_t^f\sim \sigma_t^f(\cdot|z_{1:t},a_{1:t-1}^l,x_{1:t}^f)$.
An alternative way to view the problem is as follows. As is done in the common information approach~\cite{NaMaTe13}, at time $t$, a fictitious common agent observes the common information $z_{1:t},a_{1:t-1}^l$ and generates prescription functions $\gamma_t = (\gamma_t^l,\gamma_t^f) = \psi_t[z_{1:t},a_{1:t-1}^l]$. Follower $i$ uses its prescription function $\gamma_t^{f,i}$ to operate on its private information $x_t^{f,i}$ to produce its action $a_t^{f,i}$, i.e. $\gamma_t^{f,i}:(\cX^{f,i})^t\to \cP(\cA^{f,i})$ and $a_t^{f,i}\sim\gamma_t^{f,i}(\cdot|x_{1:t}^{f,i})$. Similarly, leader uses its prescription function $\gamma_t^{l}$ to operate on its private information $x_t^{l}$ to produce its action $a_t^{l}$, i.e. $\gamma_t^{l}:(\cX^{l})^t\to \cP(\cA^{l})$ and $a_t^{l}\sim\gamma_t^{l}(\cdot|x_{1:t}^{l})$ It is easy to see that for any $\sigma$ policy profile of the players, there exists an equivalent $\psi$ profile of the common agent (and vice versa) that generates the same control actions for every realization of the information of the players.

Here, we will consider Markovian common agent's policy as follows. We call a common agent's policy be of ``type $\theta$" if the common agent observes the mean field population state $z_t$ and common belief $\pi_t$, and generates prescription functions $\gamma_t = (\gamma_t^l,\gamma_t^f) = \theta_t[\pi_t,z_t]$. 
The follower $i$ uses prescription function $\gamma_t^{f,i}$ to operate on its current private type $x_t^{f,i}$ to produce its action $a_t^{f,i}$, i.e. $\gamma_t^{f,i} : \cX^{f,i}\to \cP(\cA^{f,i})$ and $a_t^{f,i} \sim\gamma_t^{f,i}(\cdot|x_t^{f,i})$. The leader uses prescription function $\gamma_t^{l}$ to operate on its current private type $x_t^{l}$ to produce its action $a_t^{l}$, i.e. $\gamma_t^{l} : \cX^{l}\to \cP(\cA^{l})$ and $a_t^{l} \sim\gamma_t^{l}(\cdot|x_t^{l})$.

Then the mean field is updated as
\eq{z_{t+1}(\cdot) = \sum_{x_t^f,a_t}z_t(x_t^f)\pi_t(x_t^l) Q( \cdot|z_t,x_t^f, a_t^f,a_t^l)\gamma^f_t(a^f_t|x_t^f)\gamma_t^l(a_t^l|x_t^l) \label{eq:phi_def}
}
We also call the above equation as $z_{t+1} = \phi(\pi_t,z_t,\gamma_t)$

Furthermore we define a common belief $\pi_t$ on the leader's state $x_t^l$ such that $\pi_t(x_t^l) = P^{\theta}(x_t^l|z_{1:t},a_{1:t}^l)$. In the following lemma, we show that the belief $\pi_t$ can be updated using Bayes' rule.
\begin{lemma}
There exists a function $F$ independent of the strategy $\theta$ such that 
\eq{
\pi_{t+1}=F(\pi_t,z_{t},\gamma^l_t,a_t^l)
}
\end{lemma}
\begin{IEEEproof}
\eq{
\pi_{t+1}(x_{t+1}^l) &=  P^{\theta}(x_{1+t}^l|z_{1:t+1},a_{1:t}^l)\\ 
&=\frac{\displaystyle \sum_{x_t^l,a_t^l}\pi_t(x_t^l)\phi(z_{t+1}|\pi_t,z_t,\gamma_t)\gamma_t^l(a_t^l|x_t^l)Q(x_{t+1}^l|z_t,x_t^l,a_t^l)}{\displaystyle \sum_{x_t^l,a_t^l}\pi_t(x_t^l)\phi(z_{t+1}|\pi_t,z_t,\gamma_t)\gamma_t^l(a_t^l|x_t^l)}\\
&=\frac{\displaystyle \sum_{x_t^l,a_t^l}\pi_t(x_t^l)\gamma_t^l(a_t^l|x_t^l)Q(x_{t+1}^l|z_t,x_t^l,a_t^l)}{\displaystyle \sum_{x_t^l,a_t^l}\pi_t(x_t^l)\gamma_t^l(a_t^l|x_t^l)}
\label{eq:piupdate}
}
\end{IEEEproof}


In the next section, we design an algorithm to compute SMFE of the game.

\section{Algorithm for SMFE computation}
\label{sec:Result}
\subsection{Backward Recursion}
In this section, we define an equilibrium generating function $\theta=(\theta^{f,i}_t)_{i\in\{l,f\},t\in[T]}$, where $\theta^{f,i}_t :  \mathcal{P}(\cX^l)\times  \mathcal{P}(\cX^f) \to \big\{\cX^{f,i} \to \mathcal{P}(\cA^{f,i}) \big\}$ and a sequence of functions $(V_t^l,V_t^f)_{t\in \{ 1,2, \ldots T+1\}}$, where $V_t^{l}:  \mathcal{P}(\cX^l)\times \mathcal{P}(\cX^f)\times \cX^{l} \to \mathbb{R}$ and  $V_t^{f,i}:  \mathcal{P}(\cX^l)\times \mathcal{P}(\cX^f)\times \cX^{f,i} \to \mathbb{R}$, in a backward recursive way, as follows. 
\begin{itemize}
\item[1.] Initialize $\forall \pi_{T+1}\in\mathcal{P}(\mathcal{X}^l), z_{T+1}\in \mathcal{P}(\cX^f), x_{T+1}^{l}\in \cX^l, x_{T+1}^{f}\in \cX^f,$
\eq{
V^l_{T+1}(\pi_{T+1},z_{T+1},x_{T+1}^l) &\defeq 0.   \label{eq:VT+1}\\
V^f_{T+1}(\pi_{T+1},z_{T+1},x_{T+1}^f) &\defeq 0
}
\item[2.] For $t = T,T-1, \ldots 1, \ \forall \pi_t\in\mathcal{P}(X^l), z_t\in \mathcal{P}(\cX^f)$, let $\theta_t[\pi_t,z_t]$ be generated as follows. Set $\tilde{\gamma}_t = \theta_t[\pi_t,z_t]$, where $\tilde{\gamma}_t =( \tilde{\gamma}_t^l,\tilde{\gamma}_t^f)$ is the solution of the following fixed-point equation. For a given $\pi_t,z_t,\gamma_t^l$, define $\bar{BR}_t^f(\pi_t,z_t,\gamma_t^l)$ as follows, 
\seq{
  \eq{
 \bar{BR}_t^f(\pi_t,z_t,\gamma_t^l) &:=\big\{ \tgamma_t^f: \forall x_t^f\in \cX^f, \tgamma_t^f(\cdot|x_t^f)\in  \arg\max_{\gamma^f_t(\cdot|x_t^f)}\nn\\
 &\hspace{-2cm} \E^{\gamma^f_t(\cdot|x_t^f) {\gamma}^{l}_t,\,z_t,\pi_t} 
\big\{ R_t^f(z_t, X_t^l, X^f_t,A_t) +\delta V_{t+1}^f(F(\pi_t,z_t,\gamma_t^l,A_t^l),\phi(\pi_t, z_t,\gamma_t^l,\tilde{\gamma}^f_t), X^f_{t+1}) \big\lvert \pi_t, z_t,x_t^f \big\}  \big\}, \label{eq:m_FP}
}
where expectation in (\ref{eq:m_FP}) is with respect to random variables $(X_t^l,A_t,X^f_{t+1})$ through the measure\\
$\pi_t(x_t^l)\gamma^f_t(a^f_t|x^f_t) {\gamma}^{l}_t(a^{l}_t|x_t^l)$ $Q(x^f_{t+1}|x_t^l,x^f_t,a_t^l,a_t^f)$ and $\phi$ is defined in \eqref{eq:phi_def}. 

Then let for all $\pi_t,z_t$, $\theta[\pi_t,z_t] =(\tgamma_t^l,\tgamma_t^f)$ is a solution of the following fixed-point equation (if it exists),
\eq{
\tgamma_t^f &\in \bar{BR}_t^f(\pi_t,z_t,\tgamma_t^l)  \label{eq:FP1}
}
and 
\eq{
\tgamma_t^l &\in \arg\max_{\gamma_t^l}\E^{ \hat{\gamma}_t^f{\gamma}^{l}_t,\,z_t} \big\{ R_t^l(z_t,X^l_t,A_t^l) +\delta V_{t+1}^l(F(\pi_t,z_t,\gamma^l_t,A^l_t),\phi(\pi_t,z_t,\gamma_t^l,\hat{\gamma}_t^f),X_{t+1}^l)|\pi_t,z_t,x_t^l\big\},  \label{eq:FP2}\\
&\text{where } \hat{\gamma}_t^f\in \bar{BR}_t^f(\pi_t, z_t,\gamma_t^l)
}
where the above expectation is defined with respect to random variables $(X_t^l,X^f_t,A_t)$ through the measure $\pi_t(x_t^l)z_t(x^f_t)\hat{\gamma}^f_t(a^f_t|x_t^f) {\gamma}^{l}_t(a^{l}_t|x_t^l)Q(x_{t+1}^l|z_t,x_t^l,x_t^f,a_t^l,a_t^f)$, and $\hat{\gamma}^f_t\in BR_t^f(\pi_t,z_t,\gamma_t^l)$.
\label{eq:SMFE_FP}
}
Let $(\tgamma_t^l,\tgamma_t^f)$ be a pair of solution of the above operation. Then set $\forall x_t^f\in\cX^f$,
 \seq{
 \label{eq:Vdef}
  \eq{
  V^f_{t}(\pi_t,z_t,x_t^f) \defeq  &\;\E^{\tilde{\gamma}^{f}_t(\cdot|x_t) \tilde{\gamma}^{l}_t}\big\{ {R}_t^f (z_t,X^f_t,X_t^l,A_t) + \delta V_{t+1}^f (F(\pi_t,z_t,\tgamma^l_t,A_t),\phi(\pi_t,z_t,\tgamma_t^f, \tilde{\gamma}^l_t), X_{t+1}^f)\big\lvert \pi_t,z_t, x_t^f \big\}. 
  \\
   V^l_{t}(\pi_t,z_t,x_t^l) \defeq  &\;\E^{\tilde{\gamma}^{f}_t \tilde{\gamma}^{l}_t}\big\{ {R}_t^l (z_t,X^l_t,A_t^l) + \delta V_{t+1}^l (F(\pi_t,z_t,\tgamma^l_t,A_t),\phi(\pi_t,z_t,\tgamma_t^l, \tilde{\gamma}^f_t),X_{t+1}^l)\big\lvert \pi_t, z_t, x_t^l \big\}
   }
   }
   \end{itemize}

Based on $\theta$ defined in the backward recursion above, we now construct a set of strategies $\tsigma$ through forward induction as follows. 

For $t =1,2 \ldots T, \pi_t, z_{t}, x_{1:t}^f \in(\cX^f)^t,x_{1:t}^l \in(\cX^l)^t,a_{1:t-1}^l\in(\cA^l)^{t-1}$
\eq{
\pi_1(x_1^l) = Q(x_t^l)
}
\eq{
\tsigma_{t}^{f}(a_{t}^{f}|z_{1:t},a_{1:t-1}^l, x_{1:t}^{f}) &:= \theta_{t}^{f}[\pi_t,z_t](a^{f}_{t}| x_{t}^{f})\\
\tsigma_{t}^{l}(a_{t}^{l}|z_{1:t},a_{1:t-1}^l,x_{1:t}^l) &:= \theta_{t}^{l}[\pi_t,z_t](a^{l}_{t}|x_t^l)  \\
\pi_{t+1} &= F(\pi_t,z_t,\theta_t^l[\pi_t,z_t],a_t^l)\\
z_{t+1} &= \phi(\pi_t,z_t,\theta_t[\pi_t,z_t])
}

\begin{theorem}
\label{Thm:Main}
A strategy profile $\tsigma$, as constructed through backward-forward recursion algorithm above is an SMFE of the game
\end{theorem}
\begin{IEEEproof}
We will prove this theorem in two parts. In Part 1 for the follower, we prove that $\tsigma^f \in BR^f(z,\tsigma^l)$ i.e. $\ \forall \ t\in[T]$, $\forall \sigma^f, z_{1:t},a_{1:t-1}^l,x_{1:t}^f$
\eq{
 &\E^{\tsigma_{t:T}^l,\tsigma_{t:T}^f,\pi_t} \big\{ \sum_{n=t}^T \delta^{n-t}R_n^f(Z_n,X_n,A_n) |\pi_t,z_{1:t},a_{1:t-1}^l,x_{1:t}^f\big\} \geq \nn\\
 &\E^{\tsigma_{t:T}^l,\sigma_{t:T}^f,\pi_t} \big\{ \sum_{n=t}^T \delta^{n-t}R_n^f(Z_n,X_n,A_n) |\pi_t,z_{1:t},a_{1:t-1}^l,x_{1:t}^f\big\} \label{eq:Thm_f}.
}
In Part 2 for the leader, we show that $\forall z,t,\sigma^l,a_{1:t-1}^l,x_{1:t}^l$
\eq{
&\E^{\tsigma_{t:T}^l,\tsigma_{t:T}^f,\pi_t} \big\{ \sum_{n=t}^T \delta^{n-t}R_n^l(Z_n,X_n^l,A_n^l) |\pi_t,z_{1:t},a_{1:t-1}^l,x_{1:t}^l\big\} \geq\nn\\
 &\E^{\sigma^l,\hat{\sigma}^f,\pi_t} \big\{ \sum_{n=t}^T \delta^{n-t}R_n^l(Z_n,X_n^l,A_n^l) |\pi_t,z_{1:t},a_{1:t-1}^l,x_{1:t}^l\big\},\\
 &\text{where } \hat{\sigma}^f\in BR^f(z,\sigma^l)
}
where $\tsigma^f\in BR^f(z,\tsigma^l)$, as shown in Part 1.

Finally the process $z_{1:T}$ is consistent with $\tsigma^l,\tsigma^f$ such that $z=\Lambda(\tsigma^l,\tsigma^f)$. 
Combining the above parts prove the above result. The proof is presented in Appendix~C.
\end{IEEEproof}

In the following, we show that every Stackelberg mean field equilibrium can be found using the above backward recursion. This also enables us to comment on the existence of the solution of the fixed-point equation~\eqref{eq:SMFE_FP}.

\begin{theorem}
Suppose there exists an SMFE $(\tilde{\sigma}^l,\tsigma^f,z)$ that is a solution of the fixed point equation defined in Definition~\ref{def:MFE}. Then there exists an equilibrium generating function $\theta$ that satisfies \eqref{eq:SMFE_FP} in backward recursion $\forall \pi_t,z_t$
	such that  $(\tilde{\sigma}^l,\tsigma^f,z)$ is defined through forward recursion using $\theta$. This also implies that there exists a solution of~\eqref{eq:SMFE_FP} for each time $t$.
\end{theorem}
\begin{IEEEproof}
Suppose there exists an SMFE $(\tilde{\sigma}^l,\tsigma^f,z)$ of the game. The proof in Appendix~\ref{app:Proof_Exist} show that all SMFE can be found using backward/forward recursion. This proves that there exists a solution of~\eqref{eq:SMFE_FP} for every $t$.
\end{IEEEproof}

\textbf{Remark:}
When the leader is social welfare maximizing, her utility can be given by 
\eq{R^l(z_t,x_t^l,a_t^l,\gamma^f_t) = \sum_{x^f,a_t^f}z_t(x^f)\gamma_t^f(a_t^f|x_t^f)R^f(z_t,x_t^f,x_t^l,a_t^f).
}

\section{ Special case: When the leader doesn't have a private state}
In this paper, we consider the special case when the leader doesn't have a private state. The algorithm in previous section simplifies as follows.
\subsection{Algorithm for SMFE computation: Backward Recursion}
\label{s_sec:Result}
In this section, we define an equilibrium generating function $\theta=(\theta^{f,i}_t)_{i\in\{l,f\},t\in[T]}$, where $\theta^{f,i}_t :   \mathcal{P}(\cX^f) \to \big\{\cX^{f,i} \to \mathcal{P}(\cA^{f,i}) \big\}$ and a sequence of functions $(V_t^l,V_t^f)_{t\in \{ 1,2, \ldots T+1\}}$, where $V_t^{l}:  \mathcal{P}(\cX^f) \to \mathbb{R}$ and  $V_t^{f,i}:  \mathcal{P}(\cX^f)\times \cX^{f,i} \to \mathbb{R}$, in a backward recursive way, as follows. 

\begin{itemize}
\item[1.] Initialize $\forall  z_{T+1}\in \mathcal{P}(\cX^f), x_{T+1}^{f}\in \cX^f,$
\eq{
V^l_{T+1}(z_{T+1}) &\defeq 0.   \label{s_eq:VT+1}\\
V^f_{T+1}(z_{T+1},x_{T+1}^f) &\defeq 0
}
\item[2.] For $t = T,T-1, \ldots 1, \ \forall z_t\in \mathcal{P}(\cX^f)$, let $\theta_t[z_t]$ be generated as follows. Set $\tilde{\gamma}_t = \theta_t[z_t]$, where $\tilde{\gamma}_t =( \tilde{\gamma}_t^l,\tilde{\gamma}_t^f)$ is the solution of the following fixed-point equation. For a given $z_t,\gamma_t^l$, define $\bar{BR}_t^f(z_t,\gamma_t^l)$ as follows, 
  \eq{
 \bar{BR}_t^f(z_t,\gamma_t^l) &=\big\{ \tgamma_t^f: \forall x_t^f\in \cX^f, \tgamma_t^f(\cdot|x_t^f)\in  \arg\max_{\gamma^f_t(\cdot|x_t^f)}\nn\\
 &\hspace{-2cm} \E^{\gamma^f_t(\cdot|x_t^f) {\gamma}^{l}_t,\,z_t} 
\big\{ R_t^f(z_t, X^f_t,A_t) +\delta V_{t+1}^f(\phi( z_t,\gamma_t^l,\tilde{\gamma}^f_t), X^f_{t+1}) \big\lvert z_t,x_t^f \big\}  \big\}, \label{s_eq:m_FP}
}
where expectation in (\ref{s_eq:m_FP}) is with respect to random variables $(A_t,X^f_{t+1})$ through the measure\\
$\gamma^f_t(a^f_t|x^f_t) {\gamma}^{l}_t(a^{l}_t)$ $Q(x^f_{t+1}|x^f_t,a_t)$ and $\phi$ is defined in \eqref{eq:phi_def}. 

Then let for all $z_t$, $\theta[z_t] =(\tgamma_t^l,\tgamma_t^f)$ is a solution of the following fixed-point equation (if it exists),
\eq{
\tgamma_t^f &\in \bar{BR}_t^f(z_t,\tgamma_t^l)  \label{s_eq:FP1}
}
and 
\eq{
\tgamma_t^l &\in \arg\max_{\gamma_t^l}\E^{ \hat{\gamma}_t^f{\gamma}^{l}_t,\,z_t} \big\{ R_t^l(z_t,A_t^l) +\delta V_{t+1}^l(\phi(z_t,\gamma_t^l,\hat{\gamma}_t^f))|z_t\big\},  \label{s_eq:FP2}\\
&\text{where } \hat{\gamma}_t^f\in \bar{BR}_t^f( z_t,\gamma_t^l)
}
where the above expectation is defined with respect to random variables $(X^f_t,A_t)$ through the measure $z_t(x^f_t)\hat{\gamma}^f_t(a^f_t|x_t^f) {\gamma}^{l}_t(a^{l}_t)$, and $\hat{\gamma}^f_t\in BR_t^f(z_t,\gamma_t^l)$.


Let $(\tgamma_t^l,\tgamma_t^f)$ be a pair of solution of the above operation. Then set $\forall x_t^f\in\cX^f$,
  \eq{
  V^f_{t}(z_t,x_t^f) \defeq  &\;\E^{\tilde{\gamma}^{f}_t(\cdot|x_t) \tilde{\gamma}^{l}_t}\big\{ {R}_t^f (z_t,X^f_t,A_t) + \delta V_{t+1}^f (\phi(z_t,\tgamma_t^f, \tilde{\gamma}^l_t), X_{t+1}^f)\big\lvert z_t, x_t^f \big\}. 
  \label{s_eq:Vdef}\\
   V^l_{t}(z_t) \defeq  &\;\E^{\tilde{\gamma}^{f}_t \tilde{\gamma}^{l}_t}\big\{ {R}_t^l (z_t,A_t^l) + \delta V_{t+1}^l (\phi(z_t, \tilde{\gamma}^f_t,\tgamma_t^l)\big\lvert z_t \big\}
   }
   \end{itemize}

Based on $\theta$ defined in the backward recursion above, we now construct a set of strategies $\tsigma$ through  as follows. 

For $t =1,2 \ldots T,  z_{t}, x_{1:t}^f \in(\cX^f)^t$

\eq{
\tsigma_{t}^{f}(a_{t}^{f}|z_{1:t}, a_{1:t-1}^l,x_{1:t}^{f}) &:= \theta_{t}^{f}[z_t](a^{f}_{t}| x_{t}^{f})\\
\tsigma_{t}^{l}(a_{t}^{l}|z_{1:t},a_{1:t-1}^l) &:= \theta_{t}^{l}[z_t](a^{l}_{t})\\
z_{t+1} &= \phi(z_t,\theta_t[z_t])
}

\section{Numerical Examples}
 \label{sec:Example}
\subsection{Infection spread}
\label{sec:example1}

We consider a vaccination problem in a society with negative externalities. It is discretized version of the malware problem presented in~\cite{HuMa16,HUMa17,HuMa17cdc,JiAnWa11}. Some other applications of this model include security of cyberphysical systems, entry and exit of firms, investment, network effects. In this model, suppose there are large number of agents where each agent has a private state $x_t^i\in\{0,1\}$ where $x_t^i= 0$ represent `healthy' state and $x^i_t= 1$ is the infected state. Each node can take action $a_t^i\in\{0,1\}$, where $a_t^i= 0$ implies ``do nothing" and $a_t^i=1$ implies repair. The dynamics are given by
\eq{
x_{t+1}^i = \lb{x_t^i + (1-x_t^i)w_t^i\;\; \text{ for } a_t^i = 0\\
0 \hspace{63pt}\;\;\text{ for } a_t^i = 1.
}
}
where $w_t^i\in \{0,1\}$ is a binary valued random variable that represents probability of getting infected and is proportional to the population of infected nodes in the systems such that $P(w_t^i = 1) = qz(1)$. Also each infected node pays a cost $k$ and each node pays $\lambda$ to repair, 
\eq{
R^f_t(x^i_t,a^{f,i}_t,a_t^l,z_t) =-kx^i_t-c_t^l a_t^{f,i}.
}

where  $c_t^l$ is the cost of repair and $k$ represents the cost of being infected.
The leader is a social welfare maximizer that choose the cost of vaccines $c_t^l\in[0,C]$. The leader doesn't have a private type. The leader's reward is given by
\eq{
R^l(z_t,c_t^l) &= z_t(0) R^f_t(0,0,z_t)\gamma_t(0|0)+z_t(0) R^f_t(0,1,z_t)\gamma_t(1|0) +\nn\\\
&z_t(1) R^f_t(1,0,z_t)\gamma_t(0|1)+z_t(1) R^f_t(1,1,z_t)\gamma_t(1|1) +( c_t^l-c)\\
}
We pose it as an infinite horizon discounted dynamic game.
We consider parameters $k= 0.2,\delta = 0.9, q=0.9$ and $\lambda = 0.2$ for numerical results presented in Figures~4-7, and $\lambda  = 0.21$ in Figures~8-10. We note that the mean field state converges to a stationary mean-field value of 1 i.e. all agents get healthy.

\begin{figure}[htbp] 
  \centering
  \includegraphics[width=3in]{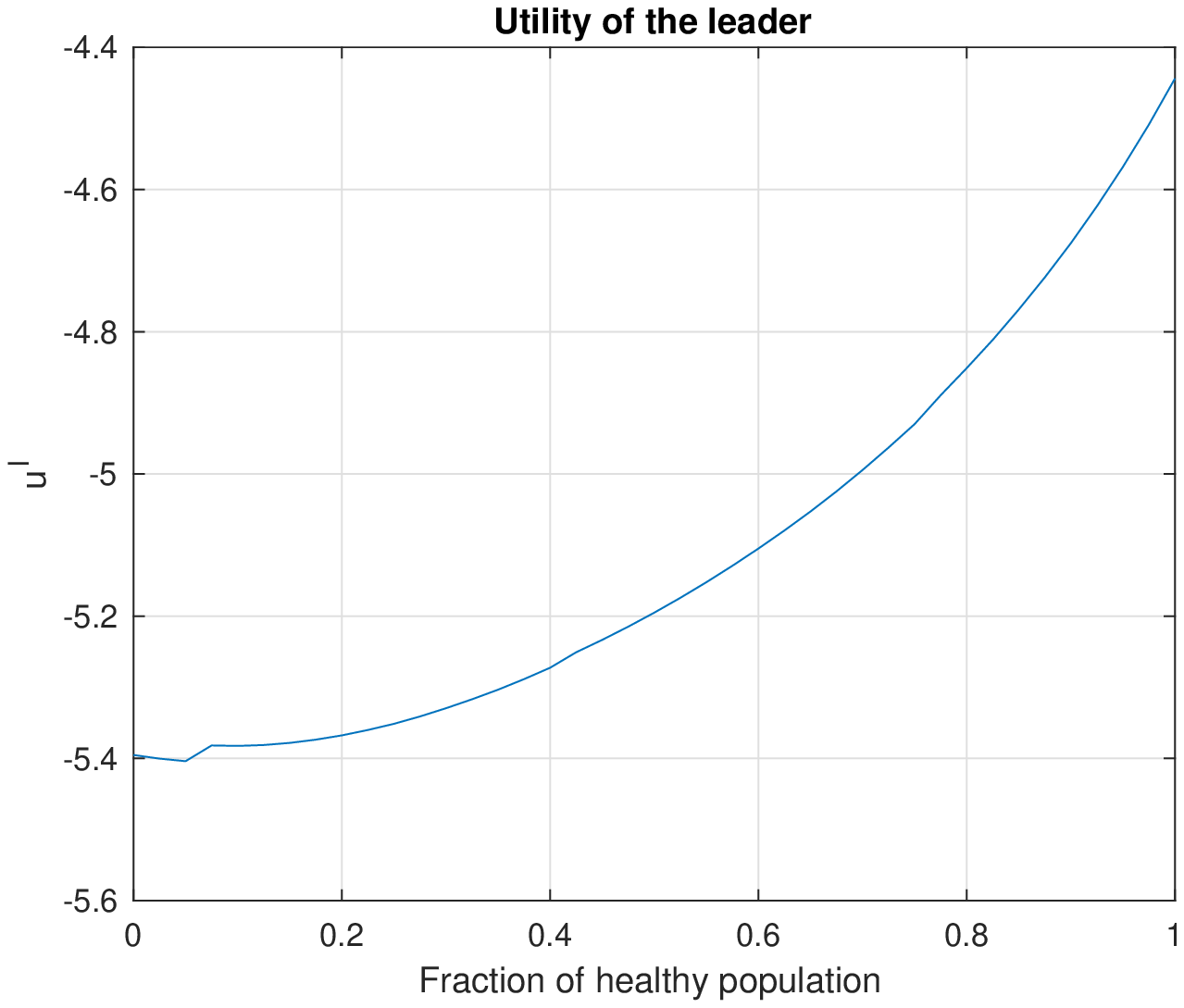} 
    \includegraphics[width=3in]{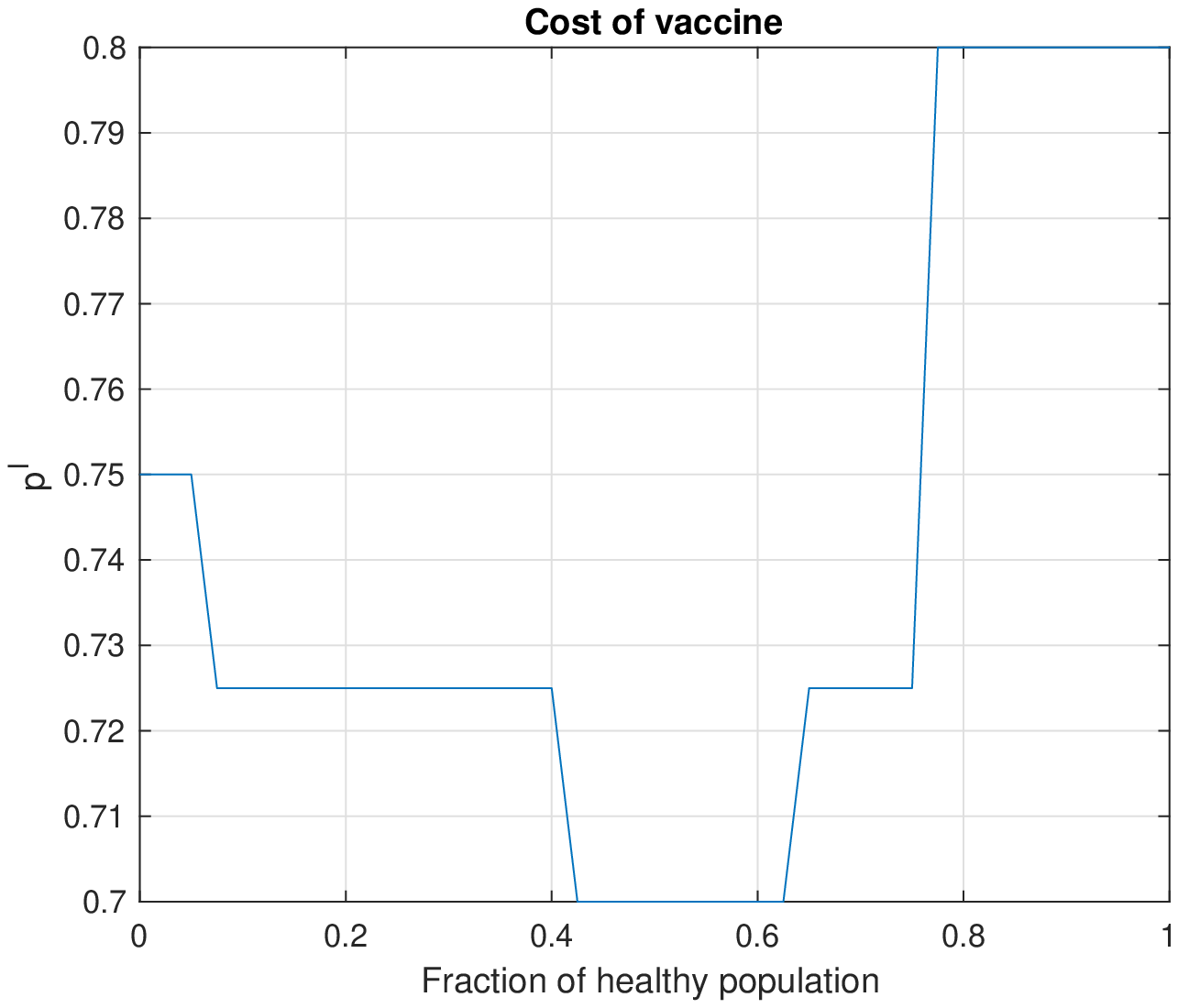} 
  \label{fig:example}
\end{figure}

\begin{figure}[htbp] 
  \centering
  \includegraphics[width=3in]{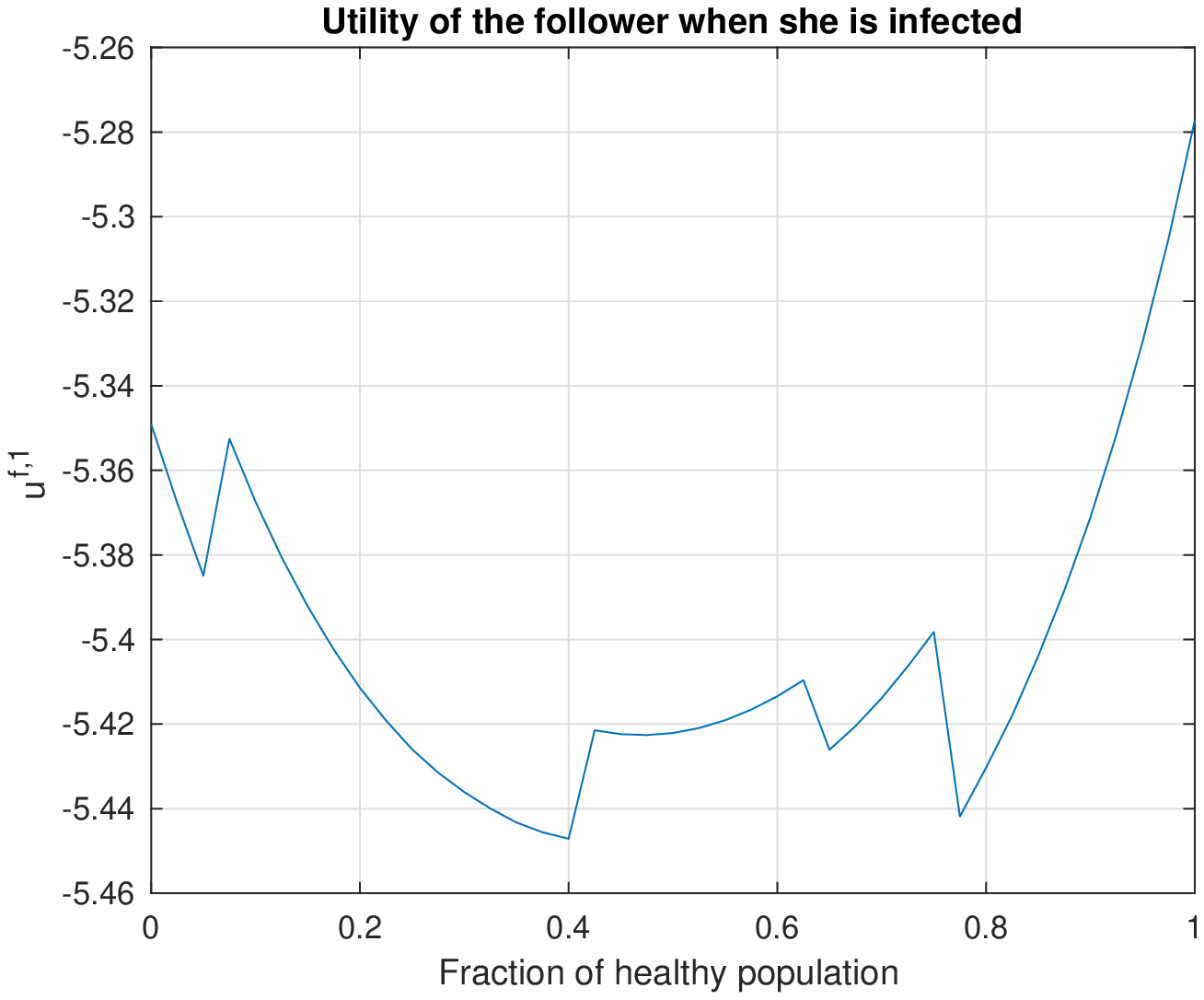} 
  \includegraphics[width=3in]{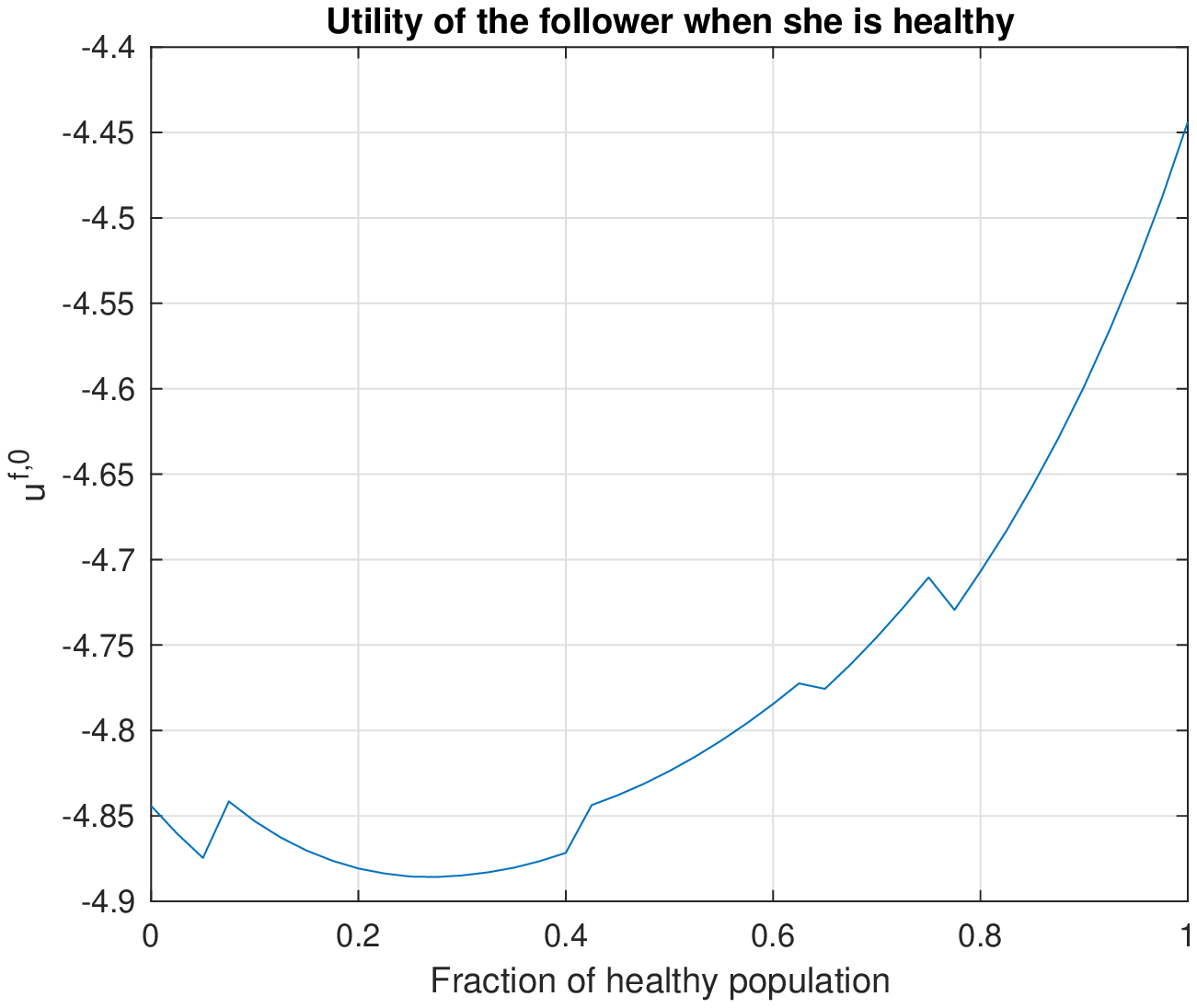}
  \label{fig:example}
\end{figure}

\begin{figure}[htbp] 
  \centering
  \includegraphics[width=3in]{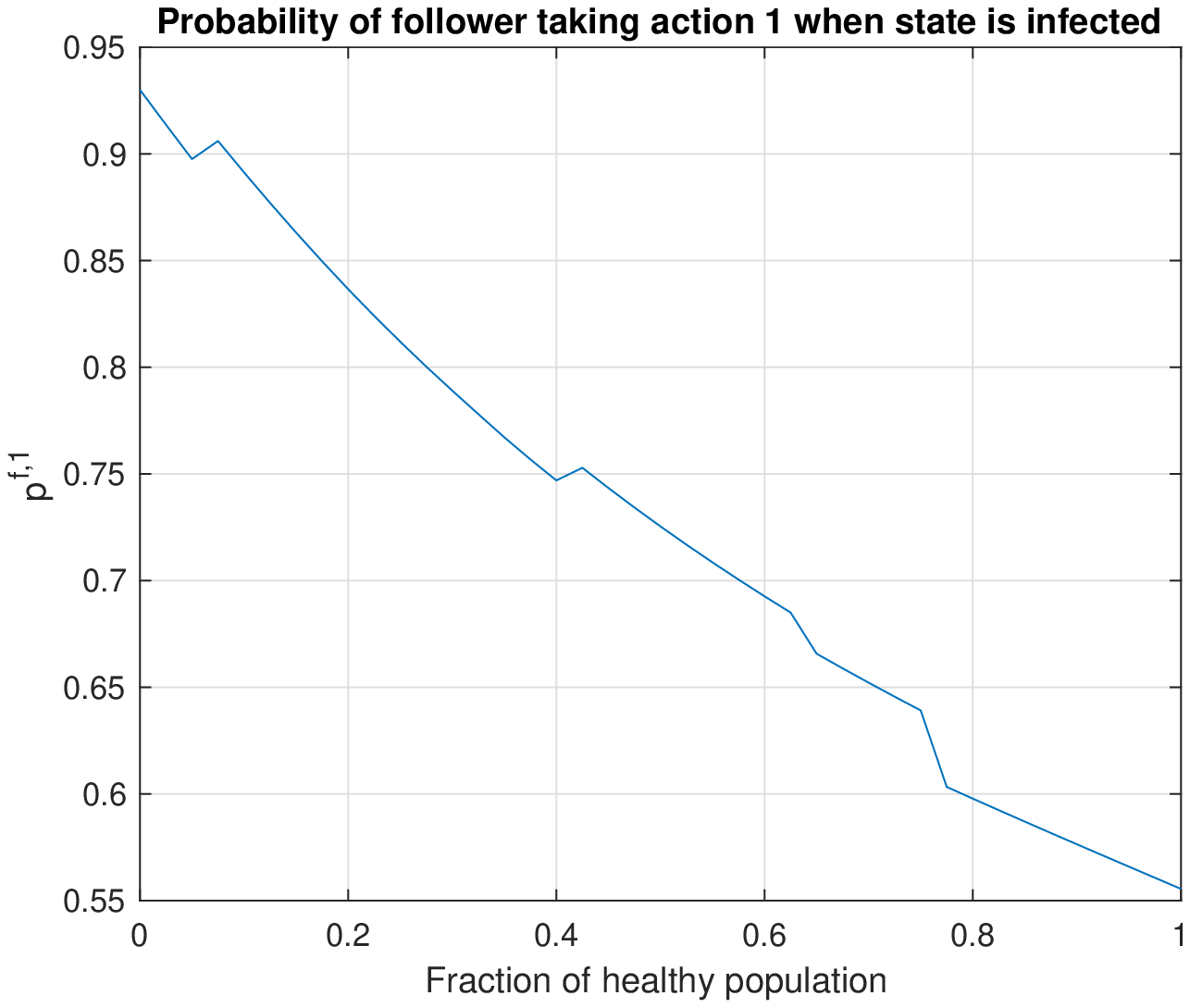} 
  \includegraphics[width=3in]{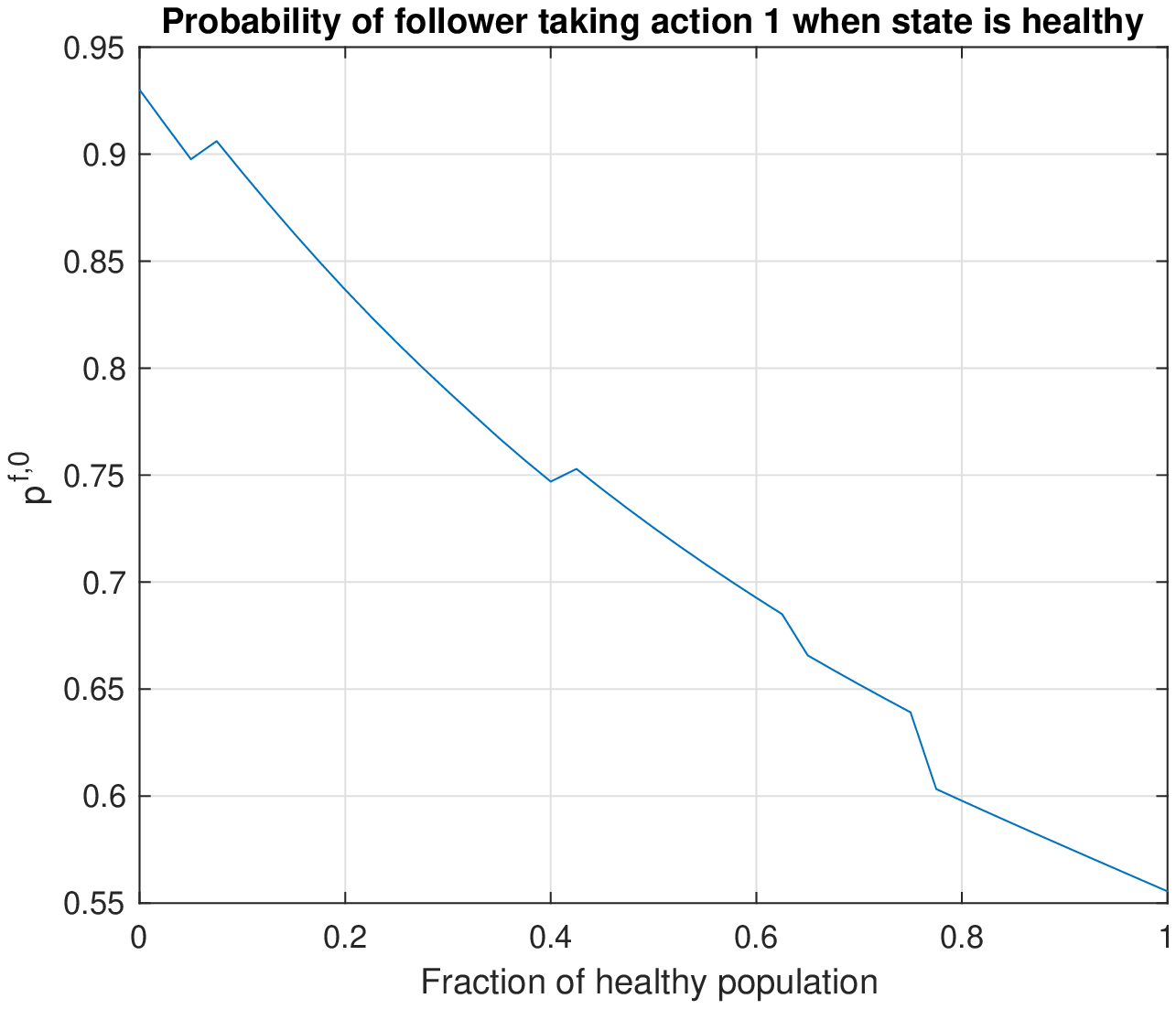}
  \label{fig:example}
\end{figure}

\begin{figure}[htbp] 
  \centering
  \includegraphics[width=3in]{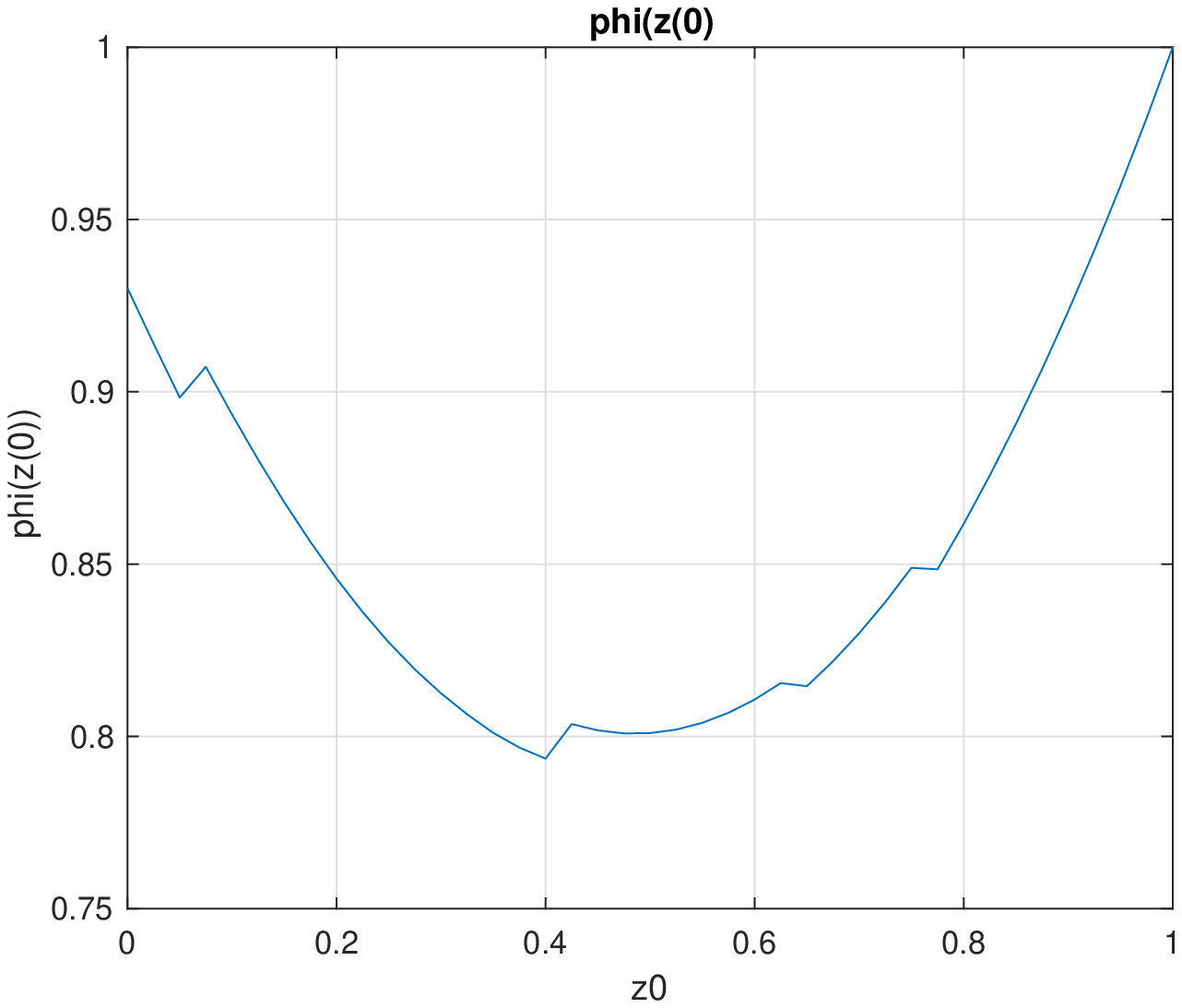} 
    \includegraphics[width=3in]{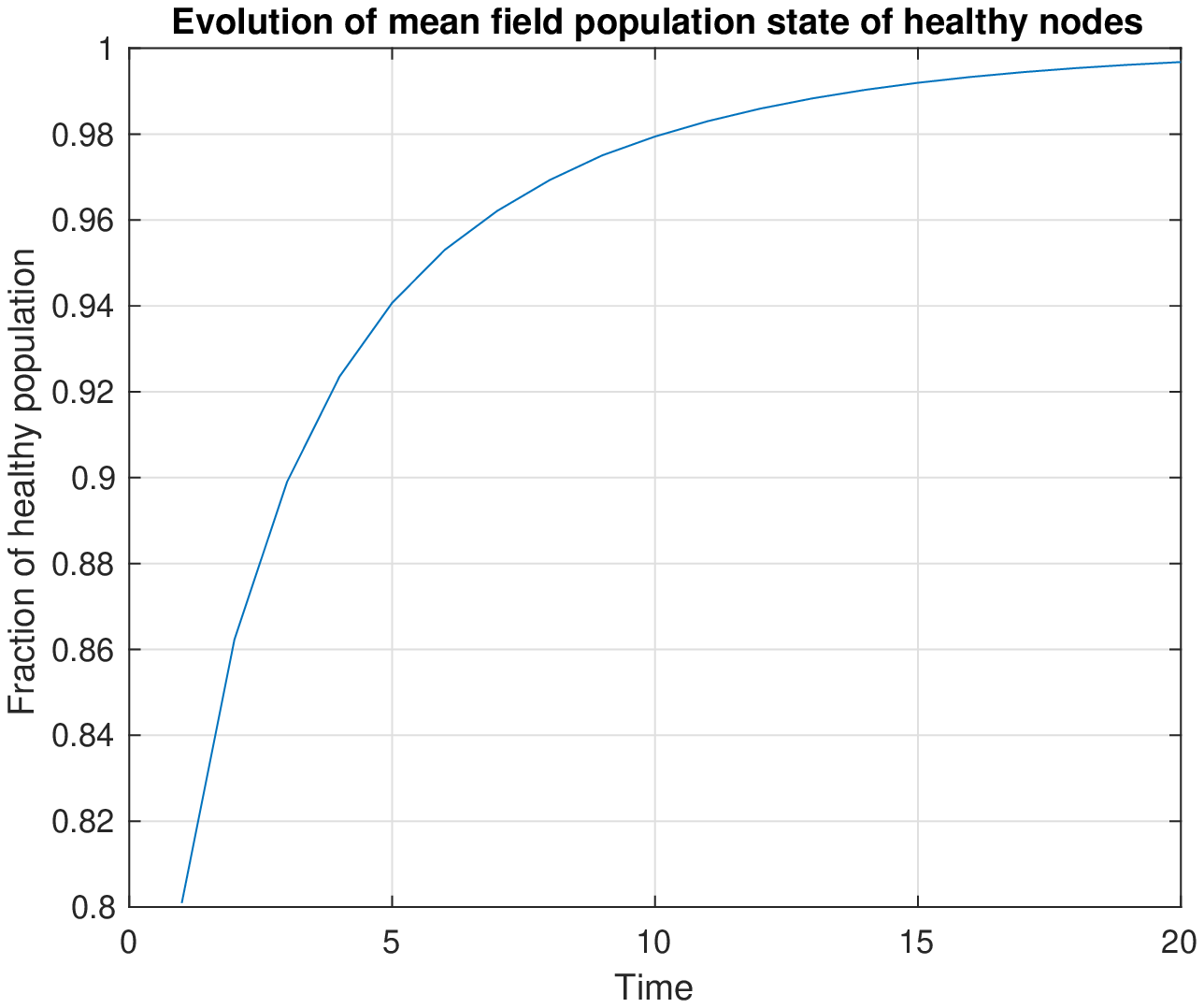} 
     \caption{$\gamma(1|0),\gamma(1|1)$: Probability of choosing repair when state $x^i=0,1$}
  \label{fig:example}
\end{figure}

\subsection{Technology adoption}

We consider a game where a firm decides the price of its product whose utility for a user depends on its personal preference and is also directly proportional to the number of other users who use that product. Each follower $i$ chooses one of the two technologies with action $a_t^{f,i} = \{-1,1\}$. She has a binary valued random variable $x_t^{f,i}\in\{-1,1\} $ that it privately observes which determines her preference for the product. The form is competing against another firm the price of whose product is fixed. The preferences of the players also evolve in an independent Markovian fashion such that
\eq{
P(x_{t+1}^{f,i}\neq x_t^{f,i}|a_t^{f,i}) =  \lb{ p^1 \text{  if  } a_t^{f,i} = x_t^{f,i}\\
 p^2 \text{  if  } a_t^{f,i} \neq x_t^{f,i},
}
}
where we assume that $p^1<p^2<1/2$. This indicates that there is some sense of ``stickiness" or inertia with the product such that if a follower has a preference for product A and choose the product A, the probability that her preference would change to the product B is lower than if she chose the product B in the first place.

Thus the mean field $z_t$ updates as
\eq{
z_{t+1}(1)&= 1- (z_t(1)\gamma_t(-1|1)Q(-1|1,-1)  +  z_t(1)\gamma_t(1|1)Q(-1|1,1)  + \nn\\
&z_t(-1)\gamma_t(-1|-1)Q(-1|-1,-1)  +  z_t(-1)\gamma_t(1|-1)Q(-1|-1,1)) \\
&=1- (z_t(1)\gamma_t(-1|1)p^2  +  z_t(1)\gamma_t(1|1)p^1  + z_t(-1)\gamma_t(-1|-1)(1-p^1)  +  z_t(-1)\gamma_t(1|-1)(1-p^2))
}
and she receives a utility 
\eq{
r_t^{f,i}(z_t,x_t^{f,i},a_t^i) = \lb{ x_t^{f,i}a_t^{f,i} + (2z_t(1)-1)a_t^{f,i} -c^1_t \text {   if    } a_t^i = 1\\
x_t^{f,i}a_t^{f,i} + (2z_t(1)-1)a_t^{f,i} -c^{-1}  \text {   if    } a_t^i = -1
}
}
where $ \{c_t^1,c_{-1}\}$ are the costs associated with actions $a_t^{f,i} = \{1,-1\}$ respectively. We assume that cost of product by the other firm is  $c^{-1}=-1$ is fixed and exogenous, however, the firm can decide the cost of the product 1.

The leader is profit maximizing and has utility
\eq{
r_t^{l}(z_t,a_t^l) = c_1\gamma_t^f(1|1)z_t(1) +c_1\gamma_t^f(1|-1)z_t(0) 
}

\begin{figure}[htbp] 
  \centering
  \includegraphics[width=3in]{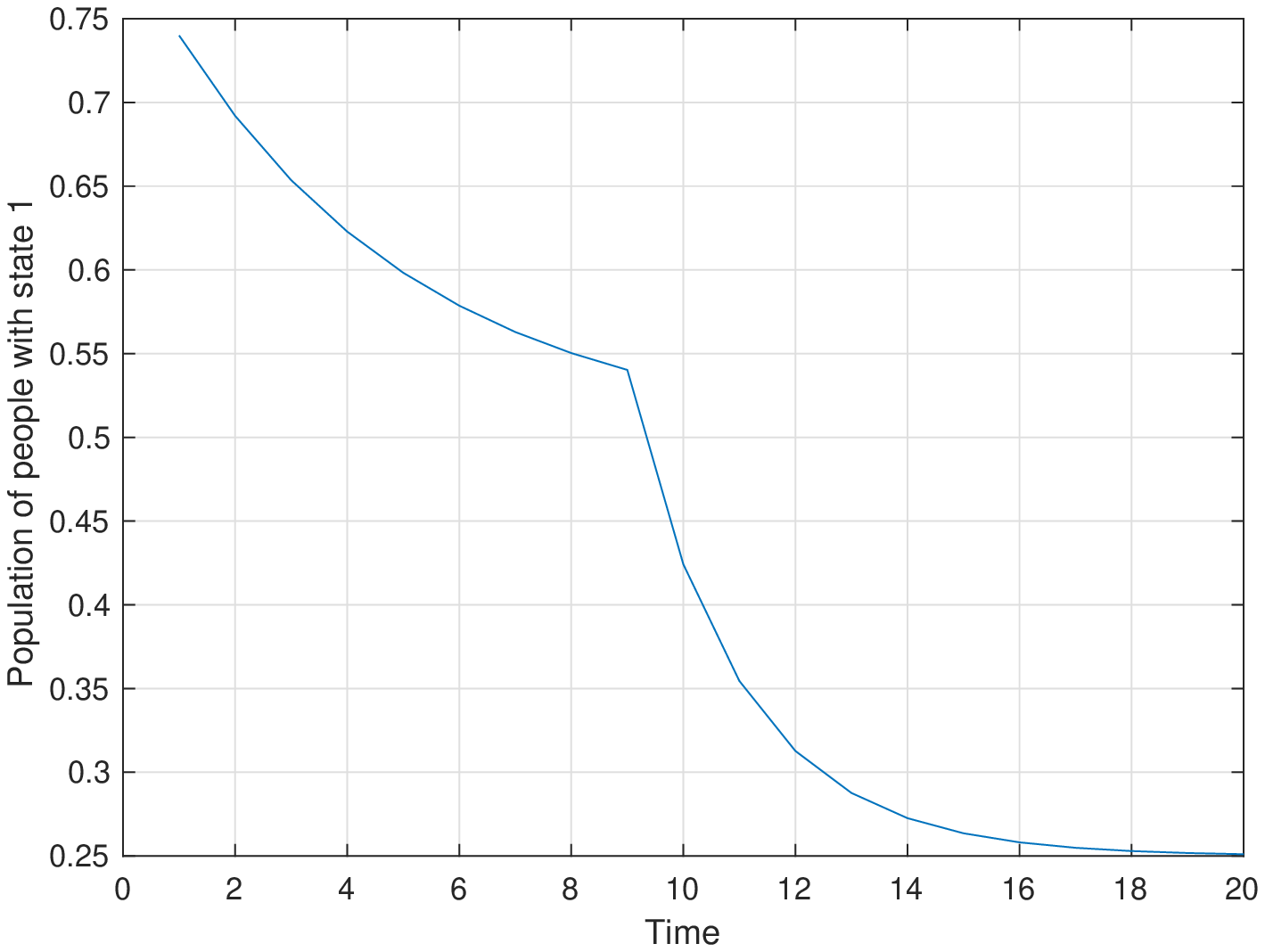} 
    \includegraphics[width=3in]{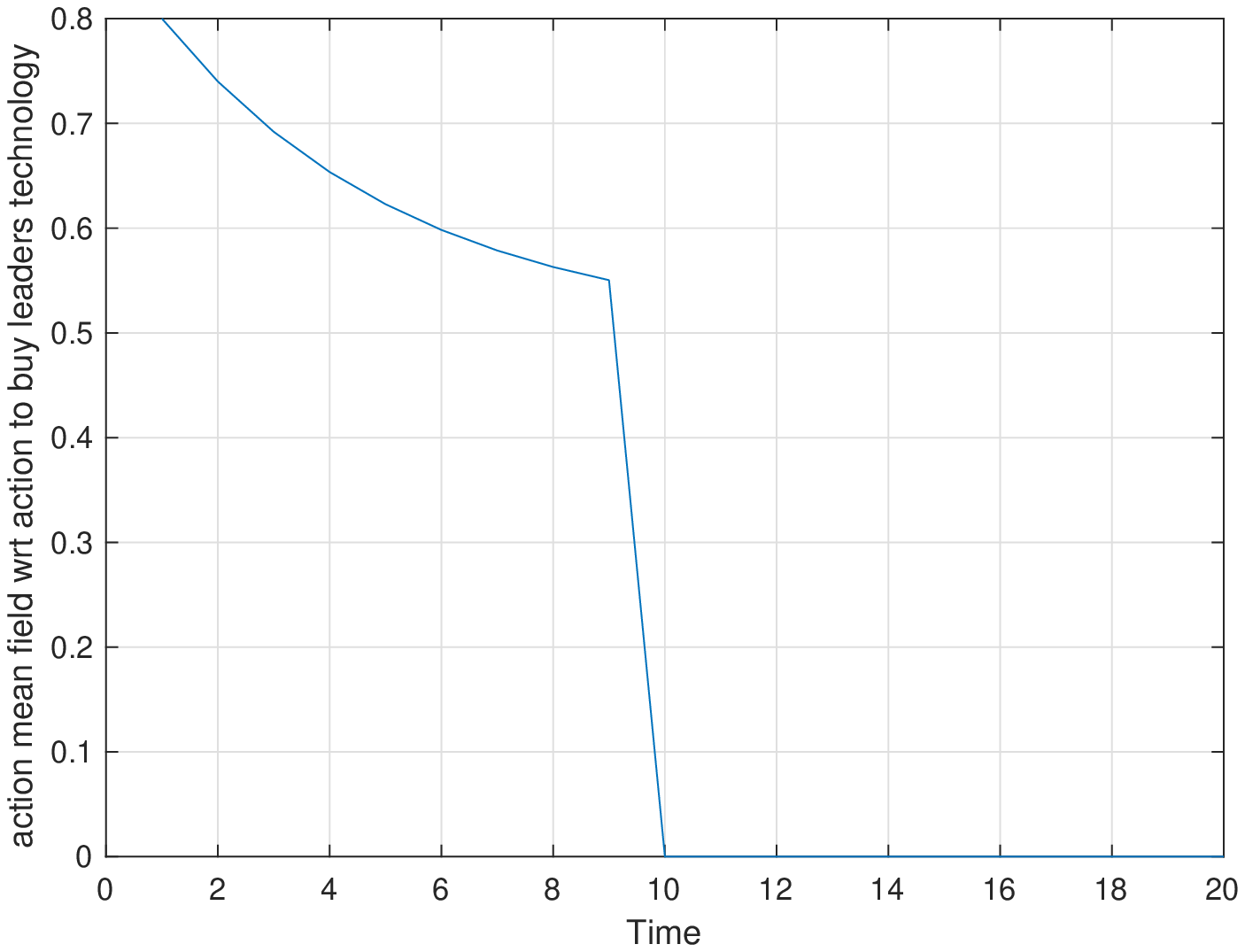} 
  \label{fig:example}
\end{figure}

\begin{figure}[htbp] 
  \centering
  \includegraphics[width=3in]{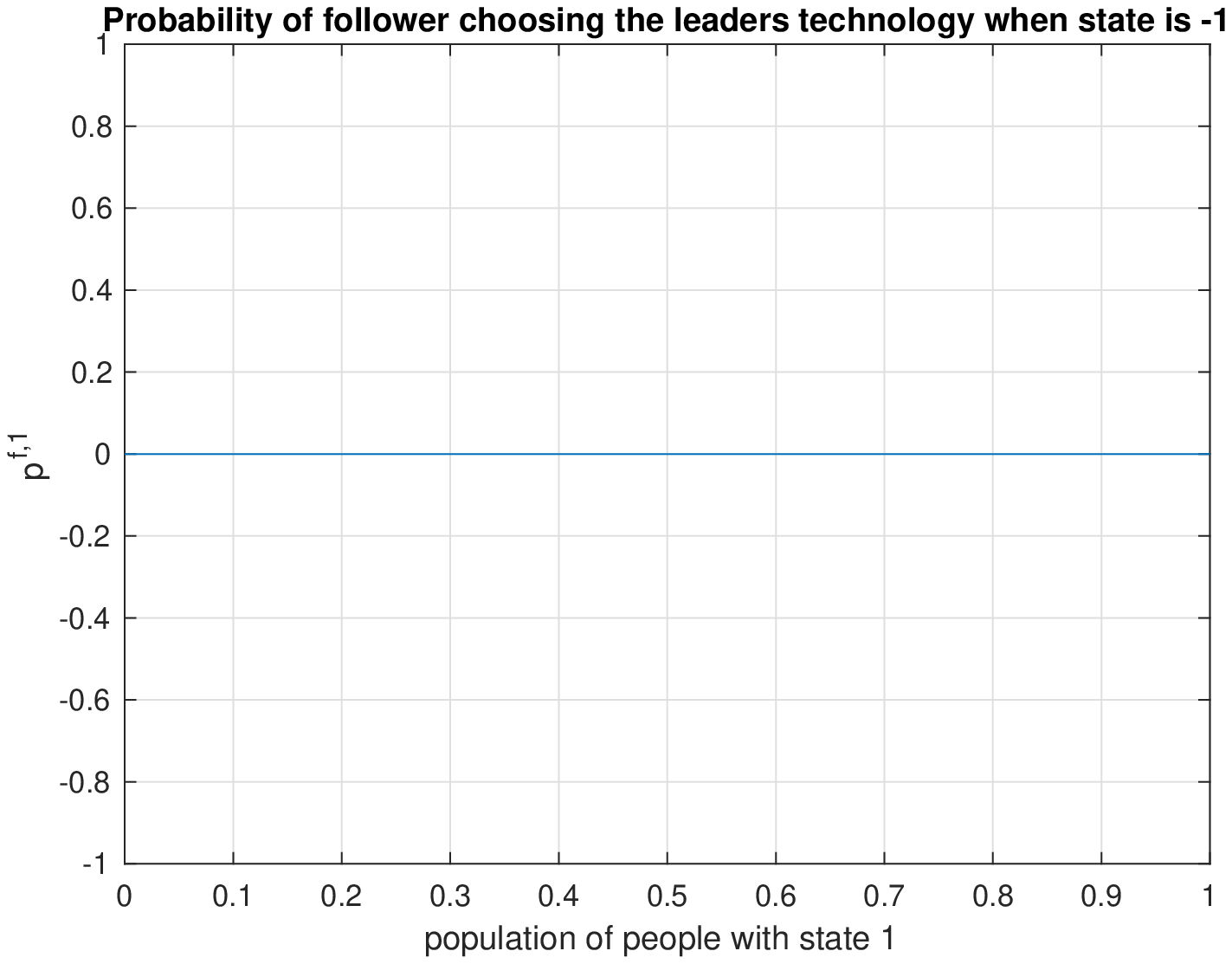} 
    \includegraphics[width=3in]{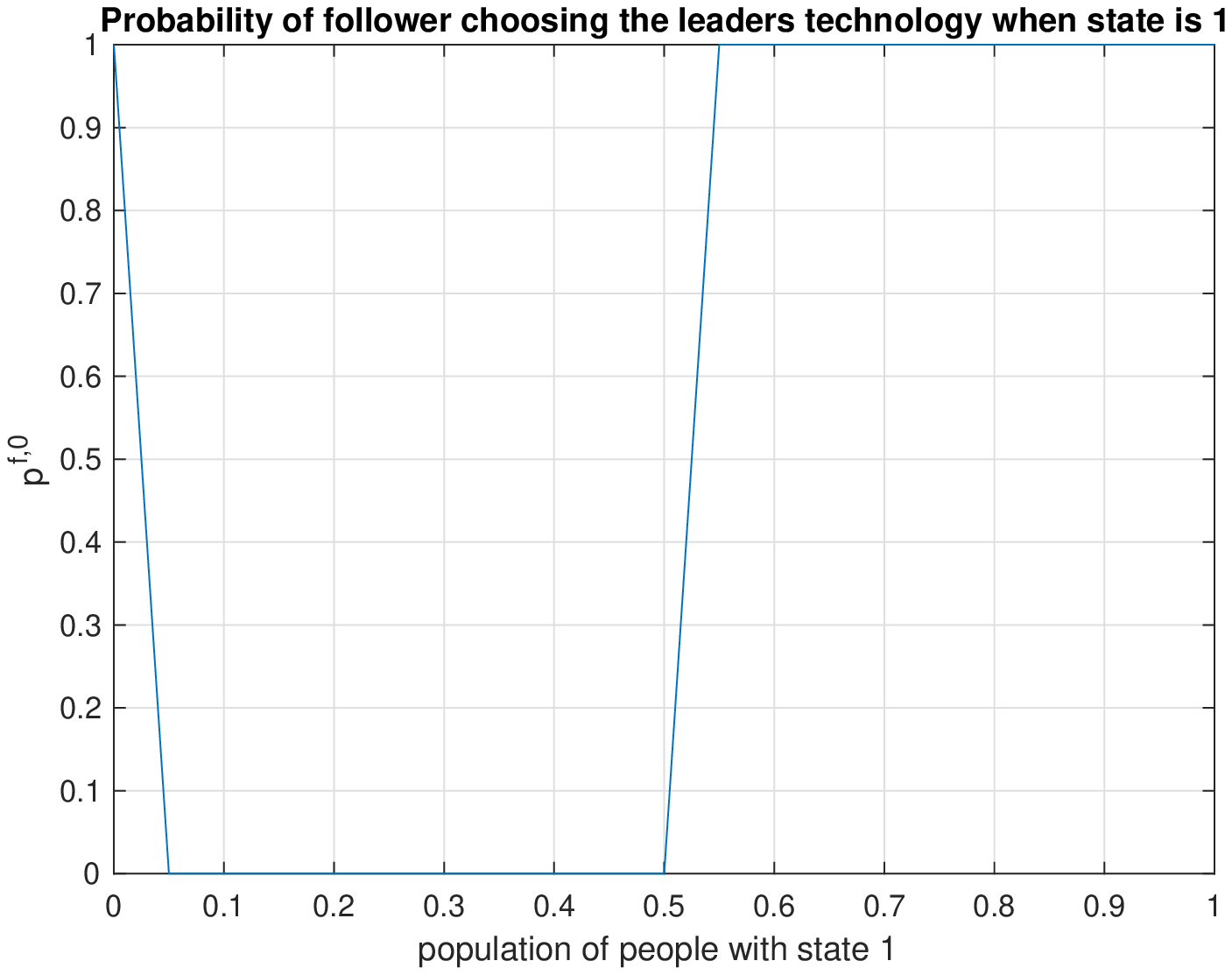} 
  \label{fig:example}
\end{figure}

\begin{figure}[htbp] 
  \centering
  \includegraphics[width=3in]{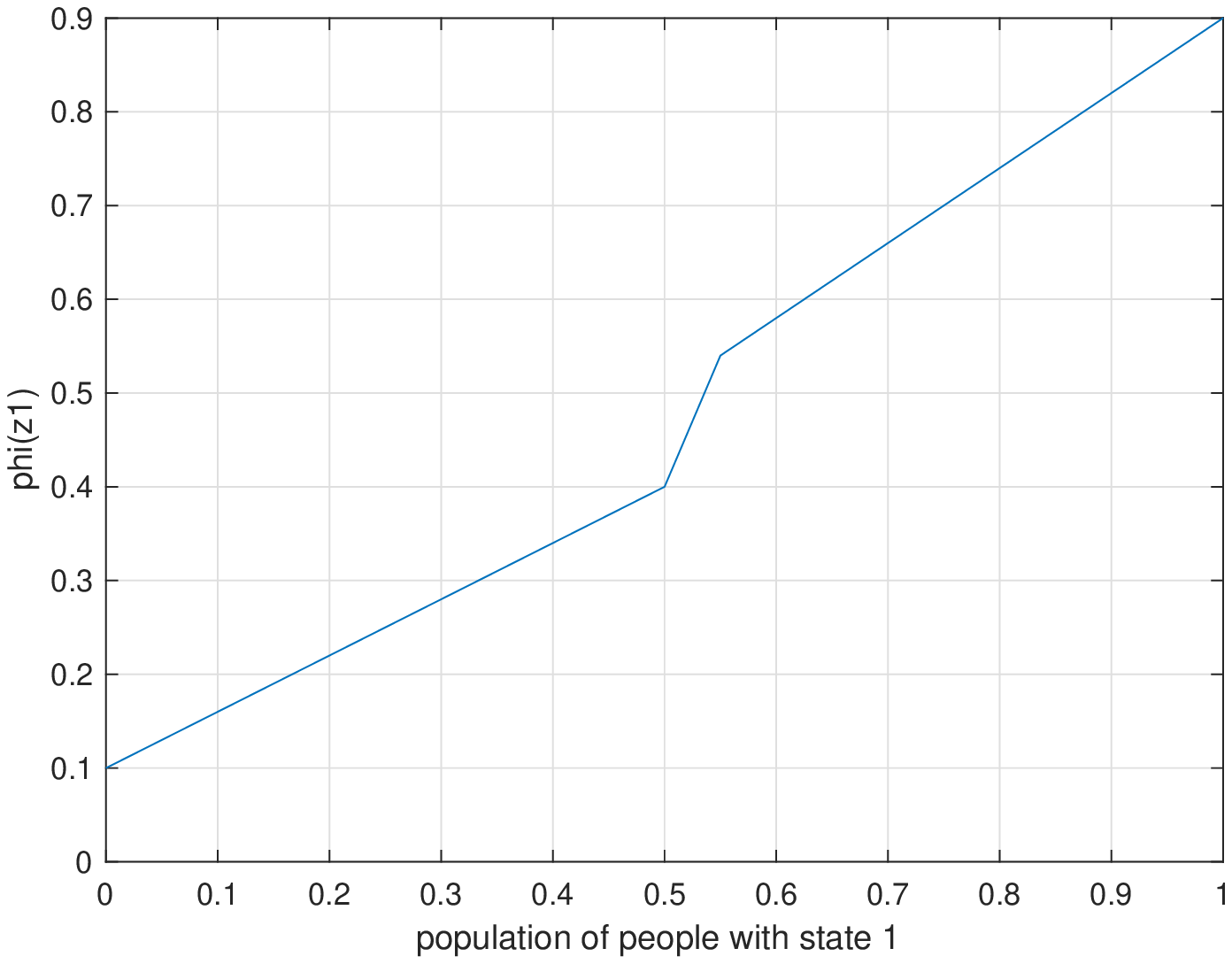} 
    \includegraphics[width=3in]{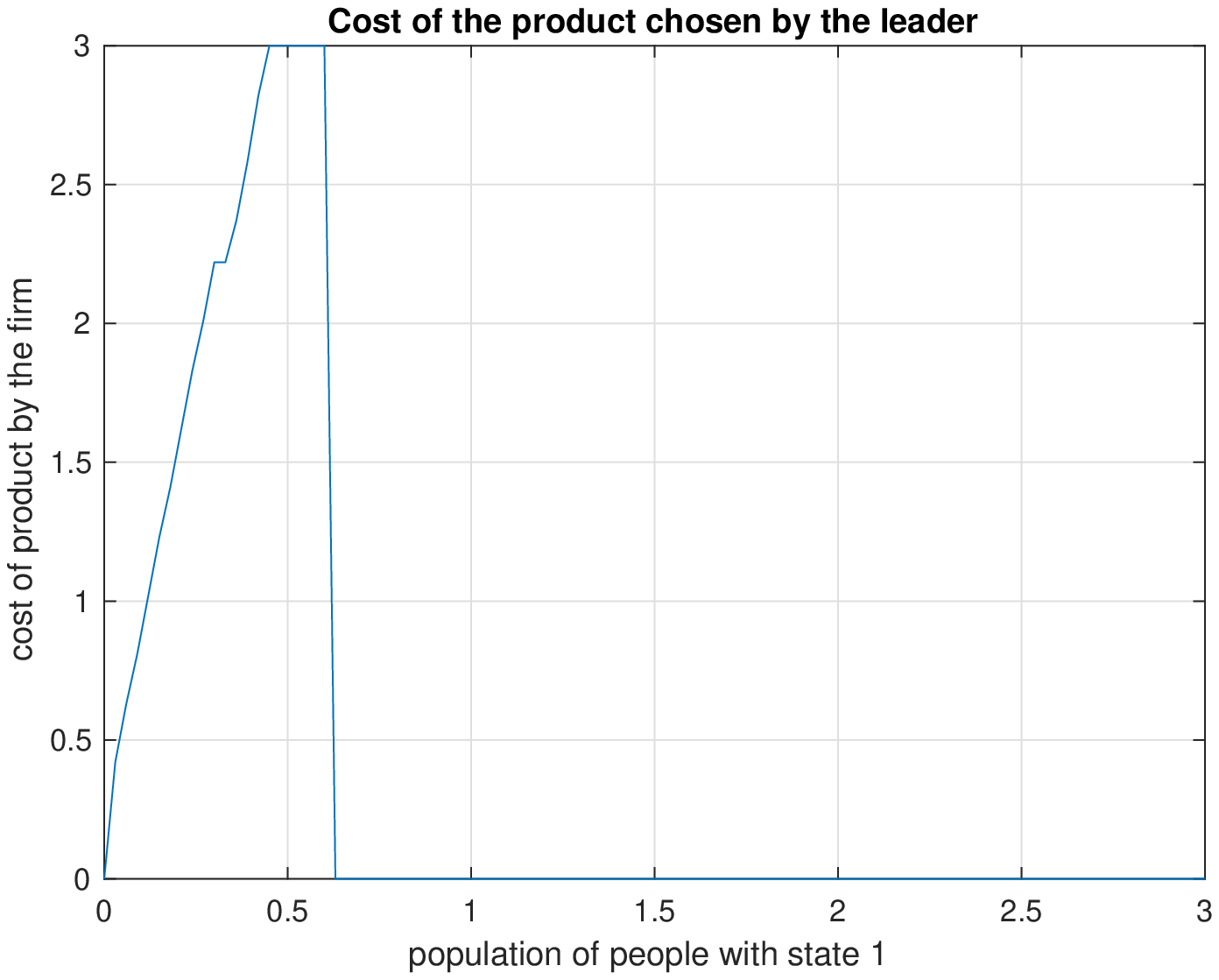} 
  \label{fig:example}
\end{figure}

\begin{figure}[htbp] 
  \centering
  \includegraphics[width=3in]{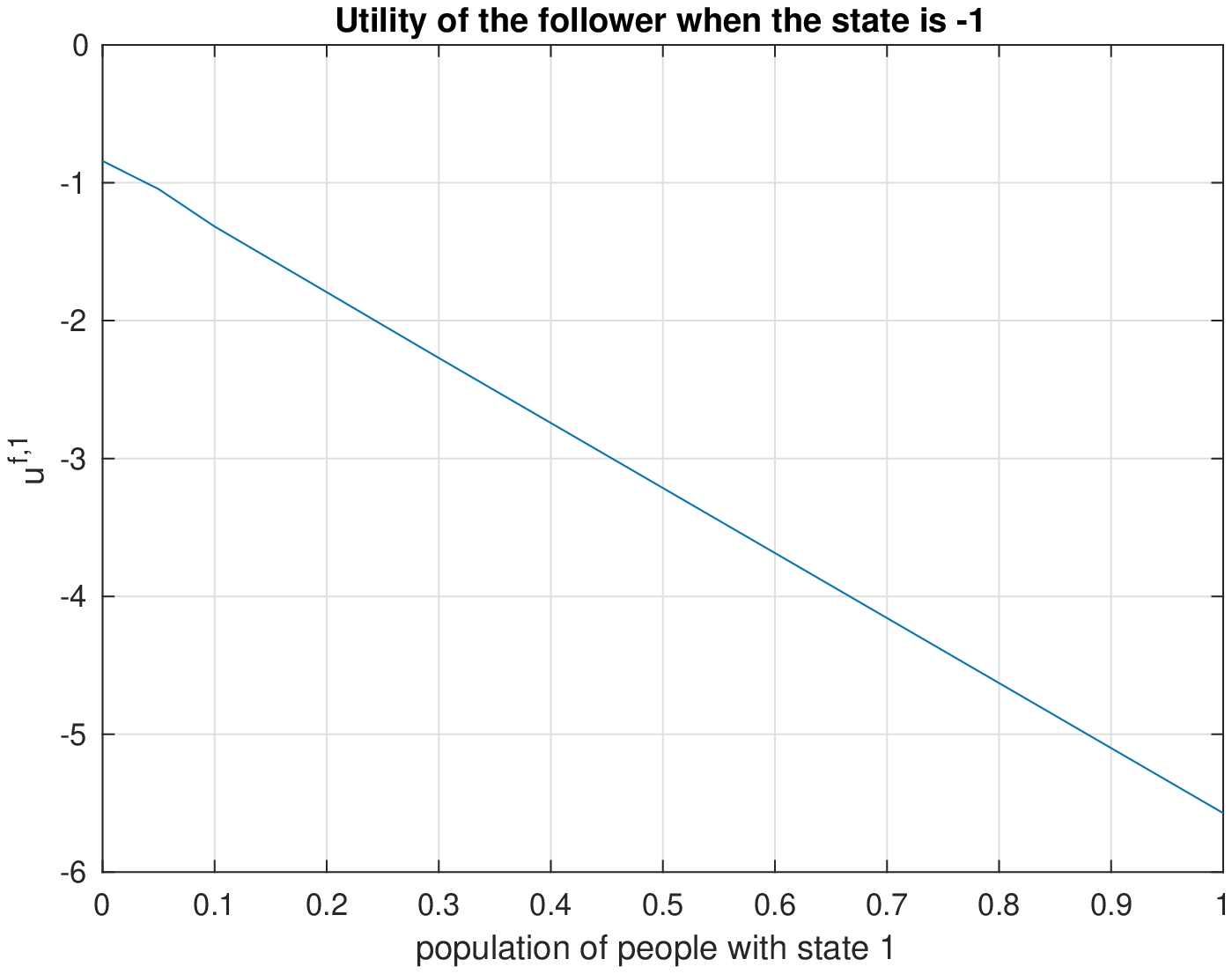} 
    \includegraphics[width=3in]{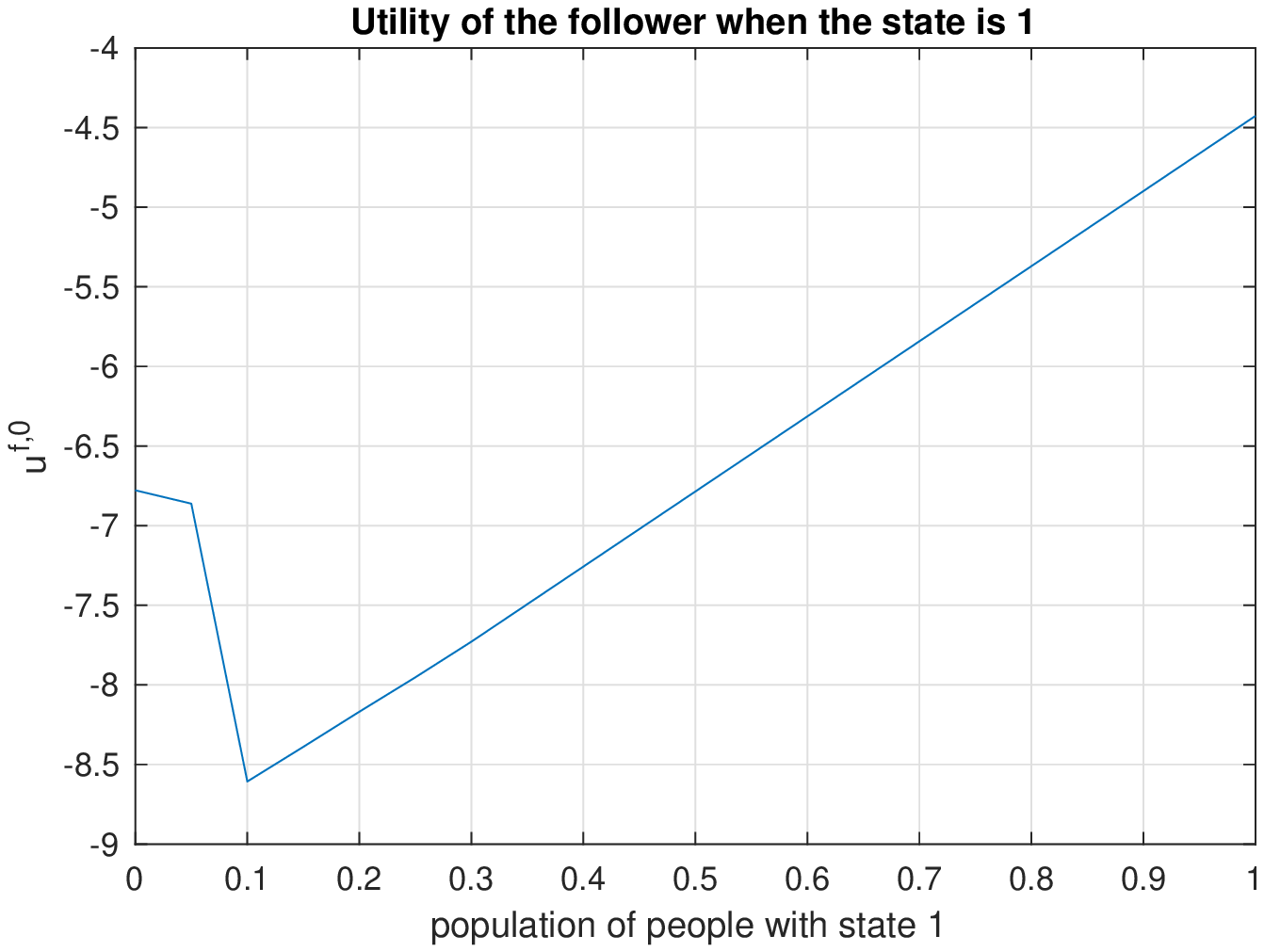} 
  \label{fig:example}
\end{figure}
\begin{figure}[htbp] 
  \centering
  \includegraphics[width=3in]{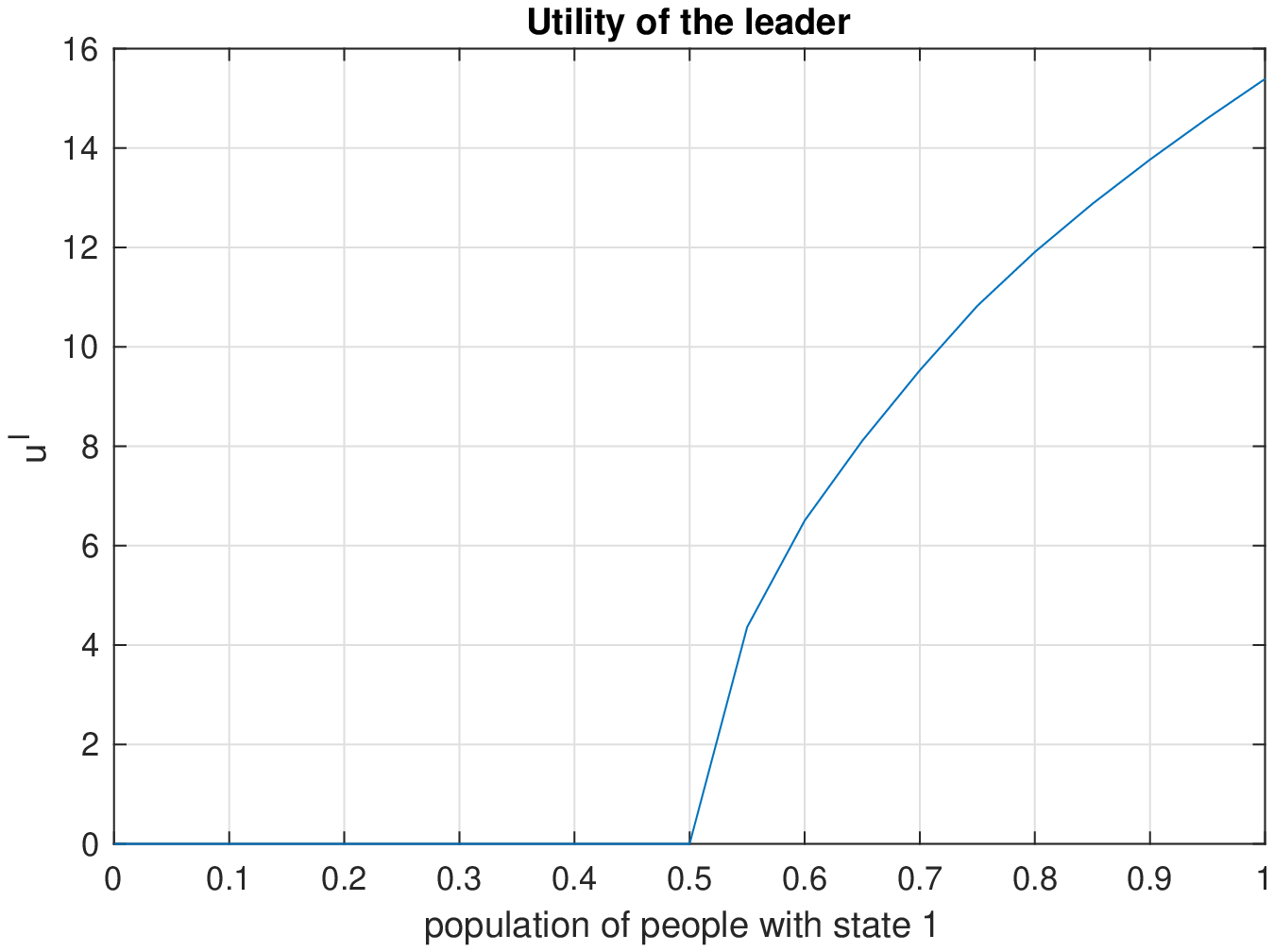} 
  \label{fig:example}
\end{figure}

Interestingly, in this game, the leader can only benefit by incentivizing to the population until the mean field converges. However, the population does eventually

\section{Conclusion}
\label{sec:Conclusion}
In this paper, we present the equivalent of Master equation for discrete time Stackelberg mean field games where the leader and the followers observe Markovian states privately and publicly observe a mean field population state. The leader commits to a dynamic policy that the followers respond to optimally. The leader, knowing that the followers will do best response, commits to a policy that maximizes her total expected reward. We define Stackelberg Mean field equilibrium (SMFE) of the game which consists of solution of a fixed-point equation across time, which consists of best response of the leader, follower and the evolution of the mean field state. We propose an algorithm to compute all SMFE of the game in a sequential manner. Based on this methodology, we numerically compute SMFE of two games: Infection spread and technology adoption. Future work involves finding sufficient conditions for the existence of the mean field equilibrium.
\appendices

\section{}
\label{app:0}
\begin{claim}
	For any policy profile $g$ and $\forall t$,
	\eq{
	\p^{\sigma}(x_{1:t}^l,x_{1:t}^f|z_{1:t},a_{1:t-1}) =  \p^{\sigma^l}(x_{1:t}^{l}|z_{1:t},a_{1:t-1})\p^{\sigma^f}(x_{1:t}^{f}|z_{1:t},a_{1:t-1})
	}
	\label{claim:CondInd}
	\end{claim}
	\begin{IEEEproof}
	\seq{
	\eq{
	&\p^{\sigma}(x_{1:t}|z_{1:t},a_{1:t-1})= \frac{\p^{\sigma}(x_{1:t},z_{1:t},a_{1:t-1})}{\sum_{\bar{x}_{1:t}} \p^{\sigma}(\bar{x}_{1:t},z_{1:t},a_{1:t-1})}
	}
	Here, we will take numerator and the denominator separately.
	\eq{
	&Nr =  \left(Q_1^l(x^l_1)\sigma^l_1(a_1^l|x_{1}^{l})\prod_{n=2}^t Q_n^l(x^l_{n}|z_{n-1},a_{n-1}^l,x^l_{n-1})\phi(z_n|\pi_{n-1},z_{n-1},\tgamma_t) \sigma^l_n(a_n^l|z_{1:n},a^l_{1:n-1},x_{1:n}^{l}) \right)\\
	&\times \left(Q_1^f(x^f_1)\sigma^f_1(a_1^f|x_{1}^{f})\prod_{n=2}^t Q_n^f(x^f_{n}|z_{n-1},a_{n-1},x^f_{n-1})\phi(z_n|\pi_{n-1},z_{n-1},\tgamma_t) \sigma^f_n(a_n^f|z_{1:n},a^l_{1:n-1},x_{1:n}^f) \right)\\
	&= \left(Q_1^l(x^l_1)\sigma^l_1(a_1^l|x_{1}^{l})\prod_{n=2}^t Q_n^l(x^l_{n}|z_{n-1},a_{n-1}^l,x^l_{n-1})\phi(z_n|\pi_{n-1},z_{n-1},\tgamma_t) \sigma^l_n(a_n^l|z_{1:n},a^l_{1:n-1},x_{1:n}^{l}) \right)\\
&\times \left(Q_1^f(x^f_1)\sigma^f_1(a_1^f|x_{1}^{f})\prod_{n=2}^t Q_n^f(x^f_{n}|z_{n-1},a_{n-1},x^f_{n-1})\phi(z_n|\pi_{n-1},z_{n-1},\tgamma_t) \sigma^f_n(a_n^f|z_{1:n},a^l_{1:n-1},x_{1:n}^{f}) \right)
	}
	and 
	\eq{
	Dr&=\sum_{x_{1:t}}\left(Q_1^l(x^l_1)\sigma^l_1(a_1^l|x_{1}^{l})\prod_{n=2}^t Q_n^l(x^l_{n}|z_{n-1},a_{n-1}^l,x^l_{n-1})\phi(z_n|\pi_{n-1},z_{n-1},\tgamma_t) \sigma^l_n(a_n^l|z_{1:n},a^l_{1:n-1},x_{1:n}^{l}) \right)\\
	&\times \left(Q_1^f(x^f_1)\sigma^f_1(a_1^f|x_{1}^{f})\prod_{n=2}^t Q_n^f(x^f_{n}|z_{n-1},a_{n-1},x^f_{n-1})\phi(z_n|\pi_{n-1},z_{n-1},\tgamma_t) \sigma^f_n(a_n^f|z_{1:n},a^l_{1:n-1},x_{1:n}^f) \right)\\
	&=\sum_{x_{1:t}^l}\left(Q_1^l(x^l_1)\sigma^l_1(a_1^l|x_{1}^{l})\prod_{n=2}^t Q_n^l(x^l_{n}|z_{n-1},a_{n-1}^l,x^l_{n-1})\phi(z_n|\pi_{n-1},z_{n-1},\tgamma_t) \sigma^l_n(a_n^l|z_{1:n},a^l_{1:n-1},x_{1:n}^{l}) \right)\\
	&\times \sum_{x_{1:t}^f}\left(Q_1^f(x^f_1)\sigma^f_1(a_1^f|x_{1}^{f})\prod_{n=2}^t Q_n^f(x^f_{n}|z_{n-1},a_{n-1},x^f_{n-1})\phi(z_n|\pi_{n-1},z_{n-1},\tgamma_t) \sigma^f_n(a_n^f|z_{1:n},a^l_{1:n-1},x_{1:n}^f) \right)\\
	}
	
	Thus
	\eq{
	\p^{\sigma}(x_{1:t}^l,x_{1:t}^f|z_{1:t},a_{1:t-1}) =  \p^{\sigma^l}(x_{1:t}^{l}|z_{1:t},a_{1:t-1})\p^{\sigma^f}(x_{1:t}^{f}|z_{1:t},a_{1:t-1})
	}
	}
	\end{IEEEproof}

\section{Part 1: Follower}
\label{app:P1}
\label{app:B}
\label{app:A}
\begin{proof}
We prove Theorem~\ref{Thm:Main} using induction and the results in Lemma~\ref{lemma:2}, and \ref{lemma:1} proved in \ref{app:B}. Let $\tsigma$ be the strategies computed by the methodology in Section~III.

\seq{
For the base case at $t=T$, $ z_{1:T},a_{1:T-1}^f,x_{1:T}^{f}, \sigma^{f}$
\eq{
\E^{\tsigma_T^l,\tsigma_{T}^{f},\pi_T}\big\{  R^f_T(Z_T,X_T,A_T) \big\lvert \pi_T,z_{1:T},a_{1:T-1}^l,x_{1:T}^f \big\}
&=
V^f_T(\pi_T,z_T, x_T^f)  \label{eq:T2a}\\
&\geq \E^{\tsigma^l,\sigma_{T}^{f},\pi_T} \big\{ R^f_T(Z_T,X_T,A_T) \big\lvert\pi_T, z_{1:T}, a_{1:T-1}^l,x_{1:T}^f \big\},  \label{eq:T2}
}
}
where \eqref{eq:T2a} follows from Lemma~\ref{lemma:1} and \eqref{eq:T2} follows from Lemma~\ref{lemma:2} in Appendix~\ref{app:B}.

Let the induction hypothesis be that for $t+1$, $\forall , z_{1:t+1}, a_{1:t}^l,x_{1:t+1}^f \in (\cX)^{t+1}, \sigma^f$,
\seq{
\eq{
 \E^{\tsigma_{t+1:T}^l,\tsigma_{t+1:T}^{f},\pi_{t+1} } \big\{ \sum_{n=t+1}^T \delta^{n-t-1}R^f_n(Z_n,X_n,A_n) \big\lvert \pi_{t+1}, z_{1:t+1},a_{1:t}^l, x_{1:t+1}^f \big\} \\
 \geq
  \E^{\tsigma_{t+1:T}^l,\sigma_{t+1:T}^{f},\pi_{t+1} } \big\{ \sum_{n=t+1}^T \delta^{n-t-1} R^f_n(Z_n,X_n,A_n) \big\lvert \pi_
  {t+1},z_{1:t+1}, a_{1:t}^l, x_{1:t+1}^f \big\}. \label{eq:PropIndHyp}
}
}
\seq{
Then $\forall z_{1:t},a_{1:t-1}^l, x_{1:t}^f, \sigma^f$, we have
\eq{
&\E^{\tsigma_{t:T}^l,\tsigma_{t:T}^{f},\pi_t } \big\{ \sum_{n=t}^T \delta^{n-t-1}R^f_n(Z_n,X_n,A_n) \big\lvert\pi_t, z_{1:t},a_{1:t-1}^l, x_{1:t}^f \big\} \nonumber \\
&= V^f_t(\pi_t,z_{t}, x_t^f)\label{eq:T1}\\
&\geq \E^{\tsigma_t^l,\sigma_t^f,\pi_t} \big\{ R^f_t(Z_t,X_t,A_t) + \delta V^f_{t+1} (F(\pi_t,z_t,\tgamma^l_t,a^l_t),\phi(\pi_t,z_{t},\tgamma_t), X_{t+1}^f) \big\lvert\pi_t, z_{1:t},a_{1:t-1}^l, x_{1:t}^f \big\}  \label{eq:T3}\\
&= \E^{\tsigma_t^l,\sigma_t^f,\pi_t } \big\{ R^f_t(Z_t,X_t,A_t) + \delta \E^{\tsigma_{t+1:T}^l,\tsigma_{t+1:T}^{f},F(\pi_t,z_t,\tgamma_t^l,a_t^l)} \nn\\
&\big\{ \sum_{n=t+1}^T \delta^{n-t-1}R^f_n(Z_n,X_n,A_n) \big\lvert F(\pi_t,z_t,\tgamma^l_t,a^l_t),  z_{1:t},Z_{t+1}, x_{1:t}^f,X_{t+1}^f \big\}  \big\vert\pi_t, z_{1:t},a_{1:t-1}^l, x_{1:t}^f \big\}  \label{eq:T3b}\\
&\geq \E^{\tsigma_t^l,\sigma_t^f,\pi_t } \big\{ R^f_t(Z_t,X_t,A_t) +  \delta\E^{\tsigma_{t+1:T}^l,\sigma_{t+1:T}^{f},F(\pi_t,z_t,\tgamma^l_t,a^l_t) } \nn\\
&\big\{ \sum_{n=t+1}^T \delta^{n-t-1}R^f_n(Z_n,X_n,A_n) \big\lvert F(\pi_t,z_t,\tgamma^l_t,a^l_t), z_{1:t},Z_{t+1}, x_{1:t}^f,X_{t+1}^f\big\} \big\vert\pi_t, z_{1:t},a_{1:t-1}^l, x_{1:t}^f \big\}  \label{eq:T4} \\
&= \E^{\tsigma_t^l,\sigma_t^f,\pi_t } \big\{ R^f_t(Z_t,X_t,A_t) +   \nn\\
&\delta\E^{\tsigma_{t:T}^l,\sigma_{t:T}^{f},\pi_t }  \big\{ \sum_{n=t+1}^T \delta^{n-t-1}R^f_n(Z_n,X_n,A_n) \big\lvert F(\pi_t,z_t,\tgamma^l_t,a^l_t), z_{1:t}, Z_{t+1}, x_{1:t}^f,X_{t+1}^f\big\} \big\vert\pi_t, z_{1:t},a_{1:t-1}^l, x_{1:t}^f \big\}  
\label{eq:T5}\\
&=\E^{\tsigma_{t:T}^l,\sigma_{t:T}^{f},\pi_t } \big\{ \sum_{n=t}^T \delta^{n-t}R^f_n(Z_n,X_n,A_n) \big\lvert\pi_t, z_{1:t},a_{1:t-1}^l,x_{1:t}^f \big\}  \label{eq:T6},
}
}
where \eqref{eq:T1} follows from Lemma~\ref{lemma:1}, \eqref{eq:T3} follows from Lemma~\ref{lemma:2}, \eqref{eq:T3b} follows from Lemma~\ref{lemma:1}, \eqref{eq:T4} follows from induction hypothesis in \eqref{eq:PropIndHyp} and \eqref{eq:T5} follows from Lemma~\ref{lemma:3}.
\end{proof}

\section{}
\label{app:B}
\label{app:lemmas}
\begin{lemma}
\label{lemma:2}
Let $\tsigma$ be the strategies computed by the methodology in Section~III. 
Then $\forall t\in [T],z_{1:t},a_{1:t-1}^l,  x_{1:t}^f, \sigma^f_t$
\eq{
V^f_t(\pi_t,z_t, x_t^f) \geq \E^{\tsigma_t^l,\sigma_t^f,\pi_t} \big\{ R^f_t(Z_t,X_t,A_t) + \delta V^f_{t+1} (F(\pi_t,z_t,\tgamma^l_t,a^l_t),\phi(\pi_t,z_{t},\tgamma_t), X_{t+1}^f) \big\lvert \pi_t, z_{1:t},a_{1:t-1}^l, x_{1:t}^f \big\}.\label{eq:lemma2}
}
\end{lemma}

\begin{proof}
We prove this lemma by contradiction.

 Suppose the claim is not true for $t$. This implies $\exists i, \hat{\sigma}_t^f, \hat{z}_{1:t},\hat{a}_{1:t-1}^l, \hat{x}_{1:t}^f$ such that
\eq{
\E^{\tsigma_t^l,\hat{\sigma}_t^f,\pi_t} \big\{ R^f_t(Z_t,X_t,A_t) +  \delta V^f_{t+1} (F(\pi_t,z_t,\tgamma^l_t,A^l_t),\phi(\pi_t,z_{t},\tgamma_t), X_{t+1}^f) \big\lvert \pi_t, \hat{z}_{1:t},\hat{a}_{1:t-1}^l,\hat{x}_{1:t}^f \big\} 
> V^f_t(\pi_t,z_t, \hat{x}_{t}^f).\label{eq:E8}
}
We will show that this leads to a contradiction.
Construct 
\begin{equation}
\hat{\gamma}^f_t(a_t^f|x_t^f) = \lb{\hat{\sigma}_t^f(a_t^f|\hat{z}_{1:t},\hat{a}_{1:t-1}^l,\hat{x}_{1:t}^f) \;\;\;\;\; x_t^f = \hat{x}_t^f \\ \text{arbitrary} \;\;\;\;\;\;\;\;\;\;\;\;\;\; \text{otherwise.}  }
\end{equation}

Then for $ \hat{z}_{1:t},\hat{a}_{1:t-1},\hat{x}_{1:t}^f$, we have
\seq{
\eq{
&V^f_t(\pi_t,z_t, \hat{x}_t^f) 
\nn \\
&= \max_{\gamma_t^f(\cdot|\hat{x}_t^f)} \E^{\tsigma^l,\gamma^f_t(\cdot|\hat{x}_t^f),\pi_t} \big\{ R^f_t(z_t,X_t^l,\hat{x}_t^f,A_t^f) + \delta V^f_{t+1} (F(\pi_t,z_t,\tgamma^l_t,A^l_t),\phi(\pi_t,z_t, \tgamma_t), X_{t+1}^f) \big\lvert \pi_t,\hat{z}_{t}, \hat{x}_{t}^f \big\}, \label{eq:E11}\\
&\geq\E^{\tsigma^l,\hat{\gamma}_t^f(\cdot|\hat{x}_t^f),\pi_t} \big\{ R^f_t(z_t,X_t,a_t) + \delta V^f_{t+1} (F(\pi_t,z_t,\tgamma^l_t,A^l_t),\phi(\pi_t,z_t, \tgamma_t), {X}_{t+1}^f) \big\lvert \pi_t, \hat{z}_t,\hat{x}_{t}^f \big\}   
\\ 
&=\sum_{x_t^l,a_t^f,x_{t+1}^f}   \big\{ R^f_t(z_t,x_t^l,\hat{x}_t^f,a_t^f) + \delta V^f_{t+1} (F(\pi_t,z_t,\tgamma^l_t,a^l_t),\phi(\pi_t,z_t, \tgamma_t), x_{t+1}^f)\big\}
\pi_t(x_t^l)\hat{\gamma}^f_t(a^f_t|\hat{x}_t^f)\nn\\
&Q_t^f(x_{t+1}^f|\hat{z}_t,\hat{x}_t^l,\hat{x}_t^f,a_t^l,a_t^f) 
\\ 
&= \sum_{x_t^l,a_t^f,x_{t+1}^f}  \big\{ R^f_t(z_t,x_t^l,\hat{x}^f_t,a_t) + \delta V^f_{t+1} (F(\pi_t,z_t,\tgamma^l_t,A^l_t),\phi(\pi_t,z_t, \tgamma_t), x_{t+1}^f)\big\}
\pi_t(x_t^l)\hat{\sigma}^f_t(a_t^f|\hat{z}_{1:t},\hat{a}_{1:t-1}^l,\hat{x}_{1:t}^f)\nn\\
&Q_t^f(x_{t+1}^f|\hat{z}_t,\hat{x}_t^l,\hat{x}_t^f,a_t^l,a_t^f) \label{eq:E9}\\
&= \E^{\tsigma^l,\hat{\sigma}_t^f,\pi_t } \big\{ R^f_t(z_t,X_t^l,\hat{x}^f_t,a_t)+ \delta V^f_{t+1} (F(\pi_t,z_t,\tgamma^l_t,A^l_t),\phi(\pi_t,z_t, \tgamma_t), X_{t+1}^f) \big\lvert \pi_t, \hat{z}_{1:t},\hat{a}_{1:t-1}^l, \hat{x}_{1:t}^f \big\}  \\
&> V^f_t(\pi_t,\hat{z}_t, \hat{x}_{t}^f), \label{eq:E10}
}
where \eqref{eq:E11} follows from definition of $V^f_t$ in \eqref{eq:Vdef}, \eqref{eq:E9} follows from definition of $\hat{\gamma}_t^f$ and \eqref{eq:E10} follows from \eqref{eq:E8}. However this leads to a contradiction.
}
\end{proof}

\begin{lemma}
\label{lemma:1}
Let $\tsigma$ be the strategies computed by the methodology in Section~III. 
Then $\forall t\in [T], z_{1:t},a_{1:t-1}^l,x_{1:t}^f$,
\begin{gather}
V^f_t(\pi_t,z_{t}, x_t^f) =
\E^{\tsigma_{t:T}^l,\tsigma_{t:T}^{f},\pi_t} \big\{ \sum_{n=t}^T \delta^{n-t}R^f_n(Z_n,X_n,A_n) \big\lvert \pi_t,  z_{1:t},a_{1:t-1}^l,x_{1:t}^f \big\} .
\end{gather} 
\end{lemma}

\begin{proof}
%
\seq{
We prove the lemma by induction. For $t=T$,
\eq{
 \E^{\tsigma_T^l,\tsigma_{T}^{f},\pi_T } \big\{  R(X_T^f,A_T^f,Z_T) \big\lvert\pi_T, z_{1:T},a_{1:T-1}^l,x_{1:T}^f \big\}
 &= \sum_{a_T^f} R^f_T(z_T,x_T^l,x_T^f,a_T) \pi_T(x_T^l)\tsigma_{T}^{f}(a_T^f|z_{T},x_{T}^f) \\
 &= V^f_T(\pi_T,z_{T}, x_T^f) \label{eq:C1},
}
}
where \eqref{eq:C1} follows from the definition of $V^f_t$ in \eqref{eq:Vdef}.
Suppose the claim is true for $t+1$, i.e., $\forall  t\in [T],  z_{1:t+1},a_{1:t}^l,x_{1:t+1}^f$
\begin{gather}
V^f_{t+1}(\pi_{t+1},z_{t+1}, x_{t+1}^f) = \E^{\tsigma_{t+1:T}^l,\tsigma_{t+1:T}^{f},\pi_{t+1}}
\big\{ \sum_{n=t+1}^T \delta^{n-t-1}R^f_n(Z_n,X_n,A_n) \big\lvert \pi_{t+1}, z_{1:t+1},a_{1:t}^l, x_{1:t+1}^f \big\} 
\label{eq:CIndHyp}.
\end{gather}
Then $\forall  t\in [T],z_{1:t},a_{1:t-1}^l, x_{1:t}^f$, we have
\seq{
\eq{
&\E^{\tsigma_{t:T}^l,\tsigma_{t:T}^{f},\pi_t } \big\{ \sum_{n=t}^T \delta^{n-t} R^f_n(Z_n,X_n,A_n) \big\lvert \pi_t, z_{1:t}, a_{1:t-1}^l,x_{1:t}^f \big\} 
\nonumber 
\\
&=  \E^{\tsigma_{t:T}^l,\tsigma_{t:T}^{f},\pi_t} \big\{R^f_t(Z_t,X_t,A_t)
\nonumber \\ 
&+\delta \E^{\tsigma_{t:T}^l,\tsigma_{t:T}^{f},\pi_t }  \big\{ \sum_{n=t+1}^T \delta^{n-t-1}R^f_n(Z_n,X_n,A_n)\big\lvert F(\pi_t,z_t,\gamma_t^l,a_t^l), z_{1:t},Z_{t+1}, x_{1:t}^f,X_{t+1}^f\big\} \big\lvert \pi_t, z_{1:t}, a_{1:t-1}^l, x_{1:t}^f \big\} \label{eq:C2}
\\
&=  \E^{\tsigma_{t:T}^l,\tsigma_{t:T}^{f},\pi_t} \big\{R^f_t(Z_t,X_t,A_t) +\delta\E^{\tsigma_{t+1:T}^l,\tsigma_{t+1:T}^{f},F(\pi_t,z_t,\tgamma^l_t,a_t^l)}
\nonumber 
\\
&\big\{ \sum_{n=t+1}^T \delta^{n-t-1}R^f_n(Z_n,X_n,A_n)\big\lvert F(\pi_t,z_t,\gamma^l_t,a^l_t), z_{1:t},Z_{t+1}, x_{1:t}^f,X_{t+1}^f\big\} \big\lvert \pi_t, z_{1:t}, a_{1:t-1}^l,x_{1:t}^f \big\} \label{eq:C3}
\\
&=  \E^{\tsigma_{t:T}^l,\tsigma_{t:T}^{f},\pi_t} \big\{R^f_t(Z_t,X_t,A_t) +  \delta V^f_{t+1}(F(\pi_t,z_t,\tgamma^l_t,A^l_t),\phi(\pi_t,z_t,\gamma_t^l,\tgamma^f_t), X_{t+1}^f) \big\lvert \pi_t,  z_{1:t},a_{1:t-1}^l, x_{1:t}^f \big\} 
\label{eq:C4}
\\
&=  \E^{\tsigma_t^l,\tsigma_{t}^{f},\pi_t} \big\{R^f_t(Z_t,X_t,A_t) +  \delta V^f_{t+1}(F(\pi_t,z_t,\tgamma^l_t,A^l_t),\phi(\pi_t,z_t,\tgamma_t), X_{t+1}^f) \big\lvert \pi_t,  z_{1:t},a_{1:t-1}^l, x_{1:t}^f \big\} 
\label{eq:C5}
\\
&=V^f_{t}(\pi_t,z_t, x_t^f) \label{eq:C6},
}
}
\eqref{eq:C4} follows from the induction hypothesis in \eqref{eq:CIndHyp} and \eqref{eq:C6} follows from the definition of $V^f_t$ in \eqref{eq:Vdef}.
\end{proof}

\begin{lemma}
\label{lemma:3}
$\forall  t\in \mathcal{T}, (z_{1:t+1},a_{1:t}^l, x_{1:t+1}^f)$ and
$\sigma^f_{t}$
\eq{
&\E^{\tsigma_{t:T}^l,  \sigma^{f}_{t:T},\,\pi_t}  \big\{ \sum_{n=t+1}^T R_n^f(Z_n,X_n,A_n) \big\lvert \pi_{t}, z_{1:t+1},a_{1:t}^l, x_{1:t+1}^f \big\} =\nn\\
& \E^{\tsigma^l_{t+1:T} \sigma^{f}_{t+1:T},F(\pi_t,z_t,\tgamma^l_t,a^l_t)}  \big\{ \sum_{n=t+1}^T R_n^f(Z_n,X_n,A_n) \big\lvert \pi_{t+1}, z_{1:t+1},a_{1:t}^l, x_{1:t+1}^f \big\}. \label{eq:F1}
}

\end{lemma}
\begin{IEEEproof} 
Since the above expectations involve random variables $X_{t+1}^{l}, Z_{t+1:T},A_{t+1:T}, X_{t+2:T}$, we consider the probability 
\seq{
\eq{
&\p^{\tsigma^l_{t:T},\sigma^f_{t:T},\,\pi_t} (x_{t+1}^l, z_{t+1:T},a_{t+1:T}, ,x_{t+2:T}\big\lvert \pi_t, z_{1:t+1}, a_{1:t}^l, x_{1:t+1}^f ) = \frac{Nr}{Dr} \label{eq:F2}
}
\vspace{-0.4cm}
\eq{
&\text{where}\nn\\ 
&Nr 
=\sum_{x_t^l,a_t^f}\p^{\tsigma^{l}_{t:T},\sigma^f_{t:T},\,\pi_t} (x_t^{l},a_t,z_{t+1}, x_{t+1}, z_{t+2:T},a_{t+1:T}, ,x_{t+2:T} \big\lvert\pi_t, z_{1:t},a_{1:t-1}^l, x_{1:t}^f ) \\
&= \sum_{x_t^{l},a_t^f}\p^{\tsigma^{l}_{t:T}, \sigma^f_{t:T},\,\pi_t} (x_t^l \big\lvert \pi_t, z_{1:t-1},a_{t+1:T}, x_{1:t}^f )\phi(z_{t+1}|\pi_{t},z_{t},\tgamma^l_t,\gamma_t^f)\sigma_t^{f}(a_t^{f}|z_{1:t},a_{1:t-1}^l, x_{1:t}^{f})\tsigma_t^{l}(a_t^{l}|\pi_t,z_{t},x_t^{l}) 
\nonumber 
\\
&Q(x_{t+1}|z_t,x_t, a_t)\p^{\tsigma^{l}_{t:T}, \sigma^f_{t:T},\,\pi_t} (z_{t+2:T}, a_{t+1:T},x_{t+2:T}| z_{1:t},a_{1:t-1}^l ,x_{1:t-1}^f, x_{t:t+1}) 
\\
=&\sum_{x_t^{l}}\pi_t(x_t^{l})\phi(z_{t+1}|\pi_{t},z_{t},\tgamma^l_t,\gamma_t^f)\sigma_t^{f}(a_t^{f}|z_{1:t},a_{1:t-1}^l, x_{1:t}^{f})  \tsigma_t^{l}(a_t^{l}|\pi_t,z_{t},x_t^{l}) Q^f(x^f_{t+1}|z_t,x^f_t, a_t)Q^{l}(x^{l}_{t+1}|z_t,x^{l}_t, a_t)
\\
&\p^{\tsigma^{l}_{t+1:T},\sigma^f_{t+1:T} ,\, \pi_{t+1}} (z_{t+2:T}, a_{t+1:T},x_{t+2:T}| \pi_t,z_{1:t} ,a_{1:t-1}^l,x_{1:t}^f,x_t^l, x_{t+1}),\label{eq:Nr2}
}
where \eqref{eq:Nr2} follows from the conditional independence of types given common information, as shown in Claim~1 in~Appendix~\ref{app:0}, and the fact that probability on $(z_{t+1:T},a_{t+1:T}, ,x_{2+t:T})$ given $z_{1:t}, x_{1:t}^f,x_t^l,x_{t+1}, \pi_{t} $ depends on $z_{1:t}, a_{1:t-1}^l,x_{1:t}^f,x_{t+1}, \pi_{t+1} $ through ${\sigma_{t+1:T}^{f}, \tsigma_{t+1:T}^{l} }$. Similarly, the denominator in \eqref{eq:F2} is given by
\eq{
Dr &= \sum_{\tilde{x}_{t}^{l},a_t^f} \p^{ \tsigma^{l}_{t:T},\sigma^f_{t:T},\, \pi_t} (\tilde{x}_t^{l}, a_t, z_{t+1},x_{t+1}^f\big\lvert \pi_t, z_{1:t-1},a_{1:t-1}^l, x_{1:t}^f )\\
&=\sum_{\tilde{x}_{t}^{l},a_t^f} \p^{\tsigma^{l}_{t:T},\sigma^f_{t:T},\,\pi_t} (\tilde{x}_t^{l} | z_{1:t-1},a_{1:t-1}^l, x_{1:t}^f )\phi(z_{t+1}|\pi_t,z_t,\tgamma_t),  \sigma_t^{f}(a_t^{f}|z_{1:t}, a_{1:t-1}^l,x_{1:t}^{f}) \nn\\
&\tsigma_t^{l}(a_t^{l}|\pi_t,z_{t}, \tilde{x}_t^{l})Q^f(x^f_{t+1}|z_t,x^f_t, a_t)\label{eq:F3}\\
=&\sum_{\tilde{x}_{t}^{l},a_t^f} \pi_t(\tilde{x}_t^{l})\phi(z_{t+1}|\pi_t,z_t,\tgamma_t) \sigma_t^{f}(a_t^{f}|z_{1:t}, a_{1:t-1}^l,x_{1:t}^{f}) \tsigma_t^{l}(a_t^{l}|\pi_t,z_{t}, \tilde{x}_t^{l})Q^f(x^f_{t+1}|z_t,x^f_t, a_t).  \label{eq:F4}
%
}

By canceling the terms $\phi(\cdot),\sigma_t^f(\cdot)$ and $Q^f(\cdot)$ in the numerator and the denominator, \eqref{eq:F2} is given by
\eq{
&\frac{\sum_{x_t^{l}}\pi_t(x_t^{l}) \tsigma_t^{l}(a_t^{l}|\pi_t,z_{t}, x_t^{l}) Q_{t+1}^{l}(x^{l}_{t+1}|z_t,x^{l}_t, a_t)}{\sum_{\tilde{x}_{t}^{l}} \pi_t^{l}(\tilde{x}_t^{l}) \tsigma_t^{l}(a_t^{l}|\pi_t,z_{t}, \tilde{x}_t^{l})}  \nonumber \\
&\times\p^{\tsigma^{l}_{t+1:T}, \sigma^f_{t+1:T} ,\, \pi_{t+1}} (z_{t+1:T}, a_{t+1:T},x_{t+2:T}|\pi_t, z_{1:t},a_{1:t-1}^l,x_{1:t}^f, x_{t+1})\\
=&\pi_{t+1}^{l}(x_{t+1}^{l}) \p^{ \tsigma^{l}_{t+1:T},\sigma^f_{t+1:T}\, \pi_{t+1}} (z_{t+1:T}, a_{t+1:T},x_{t+2:T}|\pi_t, z_{1:t} ,a_{1:t-1}^l,x_{1:t}^f, x_{t+1})\label{eq:F6}\\
=& \p^{\tsigma_{t+1:T}^{l},\sigma_{t+1:T}^{f} \, \pi_{t+1} } (x_{t+1}^{l} ,z_{t+1:T},a_{1:t}^l,x_{t+2:T} |\pi_{t+1}, z_{1:t},a_{1:t-1}^l, x_{1:t+1}^f ),
}
}
where \eqref{eq:F6} follows from using the definition of $\pi_{t+1}(x_{t+1}^{l})$ in \eqref{eq:piupdate}.

\end{IEEEproof}

\section{Part 2: Leader}
\label{app:P2}
In the following, we will show that, $t,\forall z_{1:t},a_{1:t-1}^l, x_{1:t}^l, \sigma^l$
\eq{
&\E^{\tsigma^l,\tsigma^f,\pi_t} \big\{ \sum_{n=t}^T \delta^{n-t}R_n^l(Z_n,X_n^l,A_n^l) |\pi_t,z_{1:t},a_{1:t-1}^l,x_{1:t}^l\big\} \nn\\
&\geq
 \E^{\sigma^l,\hat{\sigma}^f,\pi_t} \big\{ \sum_{n=t}^T \delta^{n-t}R_n^l(Z_n,X_n^l,A_n^l) |\pi_t,z_{1:t},a_{1:t-1}^l,x_{1:t}^l\big\},
}
where $\tsigma^f\in BR^f(z,\tsigma^l)$ as shown in Part 1 and $\hat{\sigma}^f\in BR^f(z,\sigma^l)$.

\begin{proof}
We prove the above result using induction and from results in Lemma~\ref{l_lemma:2} and \ref{l_lemma:1} proved in Appendix~\ref{l_app:lemmas}. 

For the base case at $t=T$, $\forall z_{1:T},a_{1:T-1}^l, x_{1:T}^l, \sigma^l$
\seq{
\eq{
&\hspace{-10pt}\E^{\tsigma_{T}^{f} \tsigma_{T}^{l},\pi_t}\big\{  R_T^l(Z_T,X_T^l,A_T^l) \big\lvert \pi_T, z_{1:T},a_{1:T-1}^l,x_{1:T}^l\big\}\nn \\
&\hspace{-10pt}=V^l_T(\pi_T,z_{T},x_T^l)  \label{l_eq:T2a}\\
&\hspace{-10pt}\geq \E^{\hat{\sigma}_T^f, \sigma_{T}^{l},\pi_t }\big\{ R_T^l(Z_T,X^l_T,A^l_T) \big\lvert \pi_T, z_{1:T},a_{1:T-1}^l,x_{1:T}^l \big\} \label{l_eq:T2},\nn\\
&\text{ where } \hat{\sigma}_T^f\in BR_T^f(\pi_T,z_{1:T},a_{1:T-1}^l,x_{1:T}^l,\sigma_T^l)
}
}
where (\ref{l_eq:T2a}) follows from Lemma~\ref{l_lemma:1} and (\ref{l_eq:T2}) follows from Lemma~\ref{l_lemma:2} in Appendix~\ref{l_app:lemmas}. Let the induction hypothesis be that for $t+1$, $\forall z_{1:t+1}, a_{1:t}^l,x_{1:t+1}^l, \sigma^l$,
\seq{
\eq{
 & \E^{\tsigma_{t+1:T}^{f} \tsigma_{t+1:T}^{l},\pi_{t+1}} \big\{ \sum_{n=t+1}^T R_n^l(Z_n,X_n^l,A_n^l) \big\lvert \pi_{t+1}, z_{1:t+1},a_{1:t}^l,x_{1:t+1}^l \big\} \nn\\
  &\geq \E^{\hat{\sigma}_{t+1:T}^f \sigma_{t+1:T}^{l},\pi_{t+1}} \big\{ \sum_{n=t+1}^T R_n^l(Z_n,X_n^l,A_n^l) \big\lvert \pi_{t+1}, z_{1:t+1},a_{1:t}^l,x_{1:t+1}^l \big\} \label{l_eq:PropIndHyp}\\
  &\text{where }  \hat{\sigma}^f_{t+1:T}\in BR_{t+1}^f(\pi_{t+1},z_{1:t+1},a_{1:t}^l,x_{1:t+1}^l,\sigma_{t+1:T}^l)
}
Then $\forall z_{1:t},a_{1:t-1}^l,x_{1:t}^l, \sigma^l$, we have
\eq{
&\E^{\tsigma_{t:T}^{f} \tsigma_{t:T}^{l},\pi_T} \big\{ \sum_{n=t}^T R_n^l(Z_n,X_n^l,A_n^l) \big\lvert \pi_{t}, z_{1:t},a_{1:t-1}^l,x_{1:t}^l\big\} \nn \\
&= V^l_t(\pi_t,z_{t},x_t^l)\label{l_eq:T1}\\
&\geq \E^{\hat{\gamma}^f_t, \gamma_t^{l},\pi_t} \big\{ R_t^l(Z_t,X_t^l,A_t^l) +  V_{t+1}^l (F(\pi_t,z_t,\gamma_t^l,\hat{\gamma}^f_t,A_t),\phi(\pi_t,z_t,\gamma_t^l,\hat{\gamma}_t^f),X_{t+1}^l)\big\vert \pi_{t}, z_{1:t},a_{1:t-1}^l,x_{1:t}^l \big\}  \label{l_eq:T3}\\
&= \E^{ \hat{\sigma}_t^f,\sigma_t^{l},\pi_t} \big\{ R_t^l(Z_t,X_t^l,A_t^l) +\nn\\
&\E^{\tsigma_{t+1:T}^{f} \tsigma_{t+1:T}^{l},F(\pi_t,z_t,\gamma_t^l,\hat{\gamma}^f_t,A_t)}   \big\{ \sum_{n=t+1}^T R_n^l(X_n,A_n)  \big\lvert z_{1:t},\phi(\pi_t,z_t,\gamma_t^l,\hat{\gamma}_t^f),x_{1:t}^l,X_{t+1}^l\big\}  \big\vert \pi_{t}, z_{1:t},a_{1:t-1}^l,x_{1:t}^l\big\}  \label{l_eq:T3b}\\
&\geq \E^{\hat{\sigma}_t^f, \sigma_t^{l},\pi_t} \big\{ R_t^l(Z_t,X_t^l,A_t^l) \nn\\ &+\E^{\hat{\sigma}_{t+1:T}^f \sigma_{t+1:T}^{l},F(\pi_t,z_t,\gamma_t^l,\hat{\gamma}^f_t,A_t)} \big\{ \sum_{n=t+1}^T R_n^l(Z_n,X_n^l,A_n^l)  \big\lvert z_{1:t},\phi(\pi_t,z_t,\gamma_t^l,\hat{\gamma}_t^f),x_{1:t}^l,X_{t+1}^l\big\} \big\vert \pi_{t}, z_{1:t},a_{1:t-1}^l,x_{1:t}^l \big\}  \label{l_eq:T4} \\
%
%
&= \E^{\hat{\sigma}_t^f, \sigma_t^{l},\pi_t} \big\{ R_t^l(Z_t,X_t^l,A_t^l) + \E^{\hat{\sigma}_{t:T}^{f},\sigma_{t:T}^{l},\pi_t}\nn\\
& \big\{ \sum_{n=t+1}^T R_n^l(Z_n,X_n^l,A_n^l)  \big\lvert z_{1:t},\phi(\pi_t,z_t,\sigma_t^l(\cdot|z_t,\cdot),\hat{\sigma}_t^f(\cdot|z_t,\cdot)),x_{1:t}^l,X_{t+1}^l\big\} \big\vert \pi_{t}, z_{1:t},a_{1:t-1}^l,x_{1:t}^l\big\}  \label{l_eq:T5}\\
&=\E^{\hat{\sigma}_{t:T}^f\sigma_{t:T}^{l},\pi_t} \big\{ \sum_{n=t}^T R_n^l(Z_n,X_n^l,A_n^l) \big\lvert \pi_{t}, z_{1:t},a_{1:t-1}^l,x_{1:t}^l\big\}  \label{l_eq:T6},
}
}
where $\hat{\gamma}_t^f\in \bar{BR}_t^f(\pi_t, z_t,\gamma_t^l), \hat{\sigma}_t^f \in BR^f(\pi_t,z_{1:t},a_{1:t-1}^l,x_{1:t}^l,\sigma_t^l,\tsigma_{t+1:T}^l), \hat{\gamma}_t^f = \hat{\sigma}_t^f(\cdot|z_{1:t},a_{1:t-1}^l,\cdot)$,\\ $ \hat{\sigma}^f_{t+1:T}\in BR_{t+1}^f(\pi_{t+1},z_{1:t+1},a_{1:t}^l,x_{1:t+1}^l,\sigma_{t+1:T}^l)$, (\ref{l_eq:T1}) follows from Lemma~\ref{l_lemma:1}, (\ref{l_eq:T3}) follows from Lemma~\ref{l_lemma:2}, (\ref{l_eq:T3b}) follows from Lemma~\ref{l_lemma:1} and 
(\ref{l_eq:T4}) follows from induction hypothesis in (\ref{l_eq:PropIndHyp}), (\ref{l_eq:T5}) follows from the fact that probability on $(z_{t+1:T},a^l_{t+1:T},x^l_{2+t:T})$ given $\pi_t,z_{1:t+1}, a_{1:t}^l,x_{1:t+1}^l$ depends on $\pi_{t+1},z_{1:t+1}, a_{1:t}^l,x_{1:t+1}^l$ through ${\hat{\sigma}_{t+1:T}^{f}, \tsigma_{t+1:T}^{l} }$. 
\end{proof}

\section{}
\label{l_app:lemmas}
\begin{lemma}
\label{l_lemma:2}
$\forall t\in [T], z_{1:t},a_{1:t-1}^l,x_{1:t}^l, \sigma^l_t$
\eq{
&V_t^l(\pi_t,z_{t},x_t^l) \geq \E^{\sigma_t^{l},\bar{\sigma}_t^f,\pi_t} \big\{ R_t^l(Z_t,X_t^l,A_t^l) + V_{t+1}^l (F(\pi_t,z_t,\gamma_t^l,A_t),\phi(\pi_t,z_{t},\gamma_t^l,\bar{\gamma}_t^f),X_{t+1}^l) \big\lvert \pi_{t}, z_{1:t},a_{1:t-1}^l,x_{1:t}^l\big\}\label{l_eq:lemma2}
}
where $\bar{\sigma}_t^f\in BR^f_t(\pi_t,z_{1:t},a_{1:t-1}^l,x_{1:t}^l,\sigma_t^l,\tsigma_{t+1:T}^l)$, ${\gamma}^l_t = {\sigma}_t^l(\cdot|{z}_{1:t},a_{1:t-1}^l,x_{1:t-1}^l,\cdot)$, and $\bar{\gamma}_t^f \in \bar{BR}_t^f(\pi_t,z_t,\gamma_t^l)$

\end{lemma}

\begin{proof}
We prove this lemma by contradiction. Suppose the claim is not true for $t$. This implies $\exists \hat{\sigma}^l,\hat{z}_{1:t},\hat{a}_{1:t-1}^l,\hat{x}_{1:t}^l$ such that 
\eq{
&\E^{\hat{\sigma}_t^f, \hat{\sigma}_t^{l},\pi_t} \big\{ R_t^l(Z_t,X_t^l,A_t^l) +  V_{t+1}^l (F(\pi_t,z_t,\hat{\gamma}_t^l,A_t^l),\phi(\pi_t,z_t,\hat{\gamma}_t^l,\hat{\gamma}_t^f),X_{t+1}^l) \big\lvert \pi_{t}, \hat{z}_{1:t},\hat{a}_{1:t-1}^l,\hat{x}_{1:t}^l\big\} > V_t^l(\pi_t,\hat{z}_{t},\hat{x}_t^l),\label{l_eq:E8}
}
where $\hat{\sigma}_t^f\in BR_t^f(\pi_t,z_{1:t},a_{1:t-1}^l,x_{1:t}^l,\hat{\sigma}_t^l,\tsigma_{t+1:T}^l)$ and $\hat{\gamma}_t^l$ satisfies
\eq{
\hat{\gamma}^l_t &= \hat{\sigma}_t^l(\cdot|\hat{z}_{1:t},\hat{a}_{1:t-1}^l,\hat{x}_{1:t-1}^l,\cdot)
}


Then for $\hat{z}_{1:t},\hat{a}_{1:t-1}^l,\hat{x}_{1:t}^l$, we have
\seq{
\eq{
&V_t^l(\pi_t,\hat{z}_{t},\hat{x}_t^l)\\
&= \max_{\gamma^l_t} \E^{\gamma^l_t,\breve{\gamma}_t^f,\pi_t } \big\{ R_t^l(Z_t,X_t^l,A_t^l) + V_{t+1}^l (F(\pi_t,z_t,\gamma_t^l,A_t^l),\phi(\pi_t,\hat{z}_t,\gamma_t^l,\breve{\gamma}^f_t),X_{t+1}^l) \big\lvert \pi_{t}, \hat{z}_{1:t},\hat{a}_{1:t-1}^l,\hat{x}_{1:t}^l\big\} \label{l_eq:E11}\\
&\geq\E^{\hat{\gamma}_t^l \hat{\gamma}_t^f,\pi_t} \big\{ R_t^l(Z_t,X_t^l,A_t^l)+V_{t+1}^l (F(\pi_t,z_t,\hat{\gamma}_t^l,A^l_t),\phi(\pi_t,\hat{z}_{t},\hat{\gamma}_t^l,\hat{\gamma}_t^f),X_{t+1}^l) \big\lvert \pi_{t}, \hat{z}_{1:t},\hat{a}_{1:t-1}^l,\hat{x}_{1:t}^l \big\}   \\
&= \E^{\hat{\sigma}_t^l \hat{\sigma}_t^f,\pi_t} \big\{ R_t^l(Z_t,X_t^l,A_t^l) +  V_{t+1}^l (F(\pi_t,z_t,\hat{\gamma}_t^l,A^l_t),\phi(\pi_t,\hat{z}_{t},\hat{\gamma}_t^l,\hat{\gamma}_t^f),X_{t+1}^l) \big\lvert \pi_{t}, \hat{z}_{1:t},\hat{a}_{1:t-1}^l,\hat{x}_{1:t}^l\big\}  \label{l_eq:E9}\\
&> V_t^l(\pi_t,\hat{z}_{t},\hat{x}_t^l) \label{l_eq:E10} 
}
where $\breve{\gamma} \in \bar{BR}_t^f(\pi_t,z_t,\gamma_t^l)$, (\ref{l_eq:E11}) follows from definition of $V_t^l$ in (\ref{eq:Vdef}), (\ref{l_eq:E9}) follows from definition of $\hat{\gamma}_t^l$ and (\ref{l_eq:E10}) follows from (\ref{l_eq:E8}). However this leads to a contradiction. 
}
\end{proof}

\begin{lemma}
\label{l_lemma:1}
$\forall t\in [T],z_{1:t},a_{1:t-1}^l,x_{1:t}^l$
\eq{
V^l_t(\pi_t,z_{t},x_t^l)&= \E^{\tsigma_{t:T}^{f} \tsigma_{t:T}^{l},\pi_t} \big\{ \sum_{n=t}^T R_n^l(Z_n,X_n^l,A_n^l) \big\lvert \pi_{t}, z_{1:t},a_{1:t-1}^l,x_{1:t}^l\big\} .
}
\end{lemma}
\begin{proof}
\seq{
We prove the lemma by induction. For $t=T$, 
\eq{
 &\E^{\tsigma_{T}^{f} \tsigma_{T}^{l},\pi_t} \big\{  R_T^l(Z_T,X_T^l,A_T^l) \big\lvert \pi_{T}, z_{1:T},a_{1:T-1}^l,x_{1:T}^l\big\}\nn \\
 &= \sum_{x_T^f, a_T^l} z_T(x_T^f)R_T^l(z_T,x_T^l,a_T^l) \tsigma_{T}^{f}(a_T^l|z_{T},x_T^f) \tsigma_{T}^{l}(a_T^{l}|z_{T},x_{1:T}^l)\\ 
 &=V^l_T(\pi_T,z_{T},x_T^l) \label{l_eq:C1}
}
}
where (\ref{l_eq:C1}) follows from the definition of $V_t^l$ in (\ref{eq:Vdef}).

Suppose the claim is true for $t+1$, i.e., $\forall  t\in [T], z_{1:t+1},a_{1:t}^l,x_{1:t+1}^l$
\eq{
&V^l_{t+1}(\pi_{t+1},z_{t+1},x_{t+1}^l) = \E^{\tsigma_{t+1:T}^{f} \tsigma_{t+1:T}^{l},\pi_{t+1}} \big\{ \sum_{n=t+1}^T R_n^l(Z_n,X_n^l,A_n^l) \big\lvert\pi_{t+1}, z_{1:t+1},a_{1:t}^l,x_{1:t+1}^l\big\} \label{l_eq:CIndHyp}.
}
Then $\forall  t\in [T], z_{1:t},a_{1:t-1}^l,x_{1:t}^l$, we have
	\seq{
\eq{
&\E^{\tsigma_{t:T}^{f} \tsigma_{t:T}^{l},\pi_t} \big\{ \sum_{n=t}^T R_n^l(Z_n,X_n^l,A_n^l) \big\lvert \pi_t, z_{1:t},a_{1:t-1}^l,x_{1:t}^l \big\} \nn \\
&=  \E^{\tsigma_{t:T}^{f} \tsigma_{t:T}^{l}, \pi_t } \big\{R_t^l(Z_t,X_t^l,A_t^l) +\nn\\
&\E^{\tsigma_{t:T}^{f} \tsigma_{t:T}^{l},\pi_t } \big\{ \sum_{n=t+1}^T R_n^l(Z_n,X_n^l,A_n^l) \big\lvert F(\pi_t,z_t,\tgamma^l_t,A_t^l),z_{1:t},  \phi(\pi_t,z_t,\tgamma_t),a_{1:t-1}^l,A_t^l,x_{1:t}^l,X_{t+1}^l\big\} \big\lvert \pi_t, z_{1:t},a_{1:t-1}^l,x_{1:t}^l\big\} \label{l_eq:C2}
\\
&=  \E^{\tsigma_{t:T}^{f} \tsigma_{t:T}^{l},\pi_t} \big\{R_t^l(Z_t,X_t^l,A_t^l)+\E^{\tsigma_{t+1:T}^{f} \tsigma_{t+1:T}^{l},\pi_{t+1}}\nn\\
& \big\{ \sum_{n=t+1}^T R_n^l(Z_n,X_n^l,A_n^l) \big\lvert F(\pi_t,z_t,\tgamma^l_t,A_t^l),z_{1:t},\phi(\pi_t,z_t,\tgamma_t),a_{1:t-1}^l,A_t,x_{1:t}^l,X_{t+1}^l\big\} \big\lvert \pi_t, z_{1:t},a_{1:t-1}^l,x_{1:t}^l\big\} \label{l_eq:C3}\\
&=  \E^{\tsigma_{t}^{f} \tsigma_{t}^{l},\pi_t} \big\{R_t^l(Z_t,X_t^l,A_t^l) + V^l_{t+1}(F(\pi_t,z_t,\gamma^l_t,A^l_t),\phi(\pi_t,z_t,\tgamma_t),X_{t+1}^l) \big\lvert \pi_t, z_{1:t},a_{1:t-1}^l,x_{1:t}^l\big\} \label{l_eq:C5}\\
&=V^l_{t}(\pi_t,z_t,x_t^l) \label{l_eq:C6},
}
}
where (\ref{l_eq:C3}) follows from the fact that probability on $(z_{t+1:T},a^l_{t+1:T},x^l_{2+t:T})$ given $\pi_t,z_{1:t+1}, a_{1:t}^l,x_{1:t+1}^l$ depends on $\pi_{t+1},z_{1:t+1}, a_{1:t}^l,x_{1:t+1}^l$ through ${\sigma_{t+1:T}^{f}, \tsigma_{t+1:T}^{l} }$, (\ref{l_eq:C5}) follows from the induction hypothesis in (\ref{l_eq:CIndHyp}), 
and (\ref{l_eq:C6}) follows from the definition of $V_t^l$ in (\ref{eq:Vdef}).
\end{proof}

\section{}
\label{app:Proof_Exist}

\begin{proof}
We prove this by contradiction. This implies there exists $\pi_t,z_t$ such that either (a)~\eqref{eq:FP1}
 doesn't have a solution or (b)~\eqref{eq:FP2} doesn't have a solution. 

\bit{
\item[(a)] ~\eqref{eq:FP1} doesn't have a solution (concerning the follower)

\seq{
Suppose for any equilibrium generating function $\theta$ that generates $(\tsigma^l,\tsigma^f,z)$ through forward recursion, there exists $t\in\cT, z_{1:t}, a_{1:t-1}^l$ such that for $\pi_t(\cdot)= P^{\tsigma^l,\tsigma^f}(\cdot|z_{1:t},a_{1:t-1}^l)$, \eqref{eq:FP1} is not satisfied for $\theta$
i.e. for $\tgamma^f_t = \theta^f[\pi_t,z_t] = \tsigma_t^f(\cdot|\pi_t,z_t,\cdot), \tgamma^l_t = \theta^l[\pi_t,z_t] = \tsigma_t^l(\cdot|\pi_t,z_t,\cdot)$, $\exists x_t^f$ such that
\eq{
\tgamma_t^f(\cdot|x_t^f)\notin  \arg\max_{\gamma^f_t(\cdot|x_t^f)}\nn\\
 &\hspace{-2cm} \E^{\gamma^f_t(\cdot|x_t^f) {\tgamma}^{l}_t,\,z_t,\pi_t} 
\big\{ R_t^f(z_t, X_t^l, X^f_t,A_t) +\delta V_{t+1}^f(F(\pi_t,z_t,\tgamma_t^l,A_t),\phi(\pi_t, z_t,\tgamma_t^l,\tilde{\gamma}^f_t), X^f_{t+1}) \big\lvert \pi_t, z_t,x_t^f \big\}  
  }
  Let $t$ be the first instance in the backward recursion when this happens. This implies $\exists\ \hat{\gamma}_t^f$ such that
  \eq{
  \E^{\hat{\gamma}^f_t(\cdot|x_t^f) {\tgamma}^{l}_t,\,z_t,\pi_t} 
\big\{ R_t^f(z_t, X_t^l, X^f_t,A_t) +\delta V_{t+1}^f(F(\pi_t,z_t,\tgamma_t^l,A_t^l),\phi(\pi_t, z_t,\tgamma_t^l,\tilde{\gamma}^f_t), X^f_{t+1}) \big\lvert \pi_t, z_t,x_t^f \big\}  
  \nn\\
  >  \E^{\tgamma^f_t(\cdot|x_t^f) {\tgamma}^{l}_t,\,z_t,\pi_t} 
\big\{ R_t^f(z_t, X_t^l, X^f_t,A_t) +\delta V_{t+1}^f(F(\pi_t,z_t,\tgamma_t^l,A^l_t),\phi(\pi_t, z_t,\tgamma_t^l,\tilde{\gamma}^f_t), X^f_{t+1}) \big\lvert \pi_t, z_t,x_t^f \big\}   \label{a_eq:E1}
  }
  This implies for $\hat{\sigma}^f_t(\cdot|z_{1:t},a_{1:t-1}^l,x_{1:t-1}^f,\cdot) = \hat{\gamma}_t^f$,
  \eq{
  &\E^{\tsigma_{t:T}^{f} \tsigma_{t:T}^{l},\pi_t} \big\{ \sum_{n=t}^T R_n^f(Z_n,X_n,A_n) \big\lvert \pi_t, z_{1:t},a_{1:t-1}^l, x_{1:t}^f \big\}
  \nn\\
  &= \E^{\tsigma_t^{f} \tsigma_t^{l}, \pi_t} \big\{ R_t^f(Z_t,X_t,A_t) + \E^{\tsigma_{t:T}^{f} \tsigma_{t:T}^{l},\pi_t}\nn\\
  &\big\{ \sum_{n=t+1}^T R_n^f(Z_n,X_n,A_n) \big\lvert \pi_t,z_{1:t},\phi(\pi_t,z_t,\tgamma_t),a_{1:t-1},A_t, x_{1:t+1}^f \big\}  \big\vert \pi_t, z_{1:t}, a_{1:t-1}^l, x_{1:t}^f \big\}
\\
  &= \E^{\tsigma_t^{f} \tsigma_t^{l}, \,\pi_t} \big\{ R_t^f(Z_t,X_t,A_t) + \E^{\tsigma_{t+1:T}^{f} \tsigma_{t+1:T}^{l},F(\pi_t,z_t,\tgamma_t^l,A_t^l)}\nn\\
  &\{ \sum_{n=t+1}^T R_n^f(X_n,A_n) \big\lvert z_{1:t}, \phi(\pi_t,z_t,\tgamma_t), a_{1:t-1}^l,A^l_t, x_{1:t+1}^f \big\}  \big\vert \pi_t, z_{1:t}, a_{1:t-1}^l,x_{1:t}^f \big\} \label{a_eq:E2}
  \\
  &=\E^{\tgamma^f_t(\cdot|x_t^f) \tilde{\gamma}^{l}_t, \, \pi_t} \big\{ R_t^f(Z_t,X_t,A_t) + V_{t+1}^f (F(\pi_t, z_t,\tilde{\gamma}^l_t, A^l_t),\phi(\pi_t,z_t,\tgamma_t), X_{t+1}^f) \big\lvert\pi_t,x_t^f \big\} \label{a_eq:E3}
  \\
  &< \E^{\hat{\sigma}^f_t(\cdot|\pi_t,x_t^f) \tilde{\gamma}^{l}_t, \, \pi_t} \big\{ R_t^f(Z_t,X_t,A_t) + V_{t+1}^f (F({\pi}_t, \tilde{\gamma}_t, A_t), \phi(\pi_t,z_t,\tgamma_t),X_{t+1}^f) \big\lvert \pi_t, x_t^f \big\}\label{a_eq:E4}
  \\
  &= \E^{\hat{\sigma}_t^f \sigma_t^{l}, \pi_t} \big\{ R_t^f(Z_t,X_t,A_t) +  \E^{\tsigma_{t+1:T}^{f} \tsigma_{t+1:T}^{l} \pi_{t+1}}\nn\\
  &\big\{ \sum_{n=t+1}^T R_n^f(Z_n,X_n,A_n) \big\lvert z_{1:t},\phi(\pi_t,z_t,\tgamma_t), a_{1:t-1},A_t, x_{1:t}^f,X_{t+1}^f\big\} \big\vert \pi_t,z_{1:t},a_{1:t-1}^l, x_{1:t}^f \big\}\label{a_eq:E5}
  \\
  &=\E^{\hat{\sigma}_t^f,\tsigma_{t+1:T}^{f} \tsigma_{t:T}^{l},\pi_t} \big\{ \sum_{n=t}^T R_n^f(Z_n,X_n,A_n) \big\lvert \pi_t,z_{1:t}, a_{1:t-1}^l, x_{1:t}^f \big\},\label{a_eq:E6}
  }
  where \eqref{a_eq:E2} follows from Lemma~\ref{lemma:3}, \eqref{a_eq:E3} follows from the definitions of $\tgamma_t^f$ and $\pi_t$ and Lemma~\ref{lemma:1}, \eqref{a_eq:E4} follows from \eqref{a_eq:E1} and the definition of $\hat{\sigma}_t^f$, \eqref{a_eq:E5} follows from Lemma~\ref{lemma:2}, \eqref{a_eq:E6} follows from Lemma~\ref{lemma:3}. However, this leads to a contradiction since $(\tsigma^l,\tsigma^f,z)$ is a GMFE of the game.
}

\seq{
\item[(b)] ~\eqref{eq:FP2}
 doesn't have a solution (concerning the leader)

Suppose for any equilibrium generating function $\theta$ that generates $(\tsigma^l,\tsigma^f,z)$ through forward recursion, there exists $t\in\cT, z_{1:t}, a_{1:t-1}^l$ such that for $\pi_t(\cdot)= P^{\tsigma^l,\tsigma^f}(\cdot|z_{1:t},a_{1:t-1}^l)$, \eqref{eq:FP2} is not satisfied for $\theta$
i.e. for $\tgamma^f_t = \theta^f[\pi_t,z_t] = \tsigma_t^f(\cdot|\pi_t,z_t,\cdot), \tgamma^l_t = \theta^l[\pi_t,z_t] = \tsigma_t^l(\cdot|\pi_t,z_t,\cdot)$, $\exists x_t^l$ such that
\eq{
\tgamma_t^l &\notin \arg\max_{\gamma_t^l}\E^{ \bar{\gamma}_t^f{\gamma}^{l}_t,\,z_t} \big\{ R_t^l(z_t,x^l_t,A_t^l) +\delta V_{t+1}^l(F(\pi_t,z_t,\gamma^l_t,A_t^l),\phi(\pi_t,z_t,\gamma_t^l,\bar{\gamma}_t^f),X_{t+1}^l)|\pi_t,z_t,x_t^l\big\},  \label{b_eq:FP2}\\
&\text{where } \bar{\gamma}_t^f\in \bar{BR}_t^f(\pi_t, z_t,\gamma_t^l)
} 
  Let $t$ be the first instance in the backward recursion when this happens. This implies $\exists\ \hat{\gamma}_t^l$ such that
  \eq{
  \E^{ \hat{\gamma}_t^f\hat{\gamma}^{l}_t,\,z_t} \big\{ R_t^l(z_t,x^l_t,A_t^l) +\delta V_{t+1}^l(F(\pi_t,z_t,\hat{\gamma}^l_t,A_t^l),\phi(\pi_t,z_t,\hat{\gamma}_t^l,\hat{\gamma}_t^f),X_{t+1}^l)|\pi_t,z_t,x_t^l\big\}  
  \nn\\
  >  \E^{ {\tgamma}_t^f{\tgamma}^{l}_t,\,z_t} \big\{ R_t^l(z_t,x^l_t,A_t^l) +\delta V_{t+1}^l(F(\pi_t,z_t,\tgamma^l_t,A^l_t),\phi(\pi_t,z_t,\tgamma_t^l,{\tgamma}_t^f),X_{t+1}^l)|\pi_t,z_t,x_t^l\big\}   \label{b_eq:E1}
  }
  \eq{
 &\text{where } \hat{\gamma}_t^f\in \bar{BR}_t^f(\pi_t, z_t,\hat{\gamma}_t^l)
 }
  This implies for $\hat{\sigma}^l_t(\cdot|z_{1:t},a_{1:t-1}^l,x_{1:t-1}^l,\cdot) = \hat{\gamma}_t^f$,
  \eq{
  &\E^{\tsigma_{t:T}^{f} \tsigma_{t:T}^{l},\pi_t} \big\{ \sum_{n=t}^T R_n^l(Z_n,X^l_n,A^l_n) \big\lvert \pi_t, z_{1:t},a_{1:t-1}^l, x_{1:t}^l \big\}
  \nn\\
  &= \E^{\tsigma_t^{f} \tsigma_t^{l}, \pi_t} \big\{ R_t^l(Z_t,X^l_t,A^l_t) + \E^{\tsigma_{t:T}^{f} \tsigma_{t:T}^{l},\pi_t}\nn\\
  &\big\{ \sum_{n=t+1}^T R_n^l(Z_n,X^l_n,A^l_n) \big\lvert \pi_t,z_{1:t},\phi(\pi_t,z_t,\tgamma_t),a_{1:t-1},A_t, x_{1:t+1}^l \big\}  \big\vert \pi_t, z_{1:t}, a_{1:t-1}^l, x_{1:t}^l \big\}
\\
  &= \E^{\tsigma_t^{f} \tsigma_t^{l}, \,\pi_t} \big\{ R_t^l(Z_t,X^l_t,A^l_t) + \E^{\tsigma_{t+1:T}^{f} \tsigma_{t+1:T}^{l},F(\pi_t,z_t,\tgamma_t,A_t)}\nn\\
  &\{ \sum_{n=t+1}^T R_n^l(Z_n,X^l_n,A^l_n) \big\lvert \pi_t, z_{1:t}, \phi(\pi_t,z_t,\tgamma_t), a_{1:t-1},A_t, x_{1:t+1}^l \big\}  \big\vert \pi_t, z_{1:t}, a_{1:t-1}^l,x_{1:t}^l \big\} \label{b_eq:E2}
  \\
  &=\E^{\tgamma^f_t(\cdot|x_t^f) \tilde{\gamma}^{l}_t, \, \pi_t} \big\{ R_t^l(Z_t,X^l_t,A^l_t) + V_{t+1}^l (F(\pi_t,z_t, \tilde{\gamma}^l_t, A^l_t),\phi(\pi_t,z_t,\tgamma_t), X_{t+1}^l) \big\lvert\pi_t,x_t^l \big\} \label{b_eq:E3}
  \\
  &< \E^{\hat{\sigma}^f_t(\cdot|\pi_t,x_t^f) \tilde{\gamma}^{l}_t, \, \pi_t} \big\{ R_t^l(Z_t,X_t,A_t) + V_{t+1}^l (F({\pi}_t, \tilde{\gamma}^l_t, A^l_t), \phi(\pi_t,z_t,\tgamma_t),X_{t+1}^l) \big\lvert \pi_t, x_t^l \big\}\label{b_eq:E4}
  \\
  &= \E^{\hat{\sigma}_t^f \sigma_t^{l}, \pi_t} \big\{ R_t^l(Z_t,X^l_t,A^l_t) +  \E^{\tsigma_{t+1:T}^{f} \tsigma_{t+1:T}^{l} \pi_{t+1}}\nn\\
  &\big\{ \sum_{n=t+1}^T R_n^l(Z_n,X^l_n,A^l_n) \big\lvert z_{1:t},\phi(\pi_t,z_t,\tgamma_t), a_{1:t-1},A_t, x_{1:t}^l,X_{t+1}^l\big\} \big\vert \pi_t,z_{1:t},a_{1:t-1}^l, x_{1:t}^l \big\}\label{b_eq:E5}
  \\
  &=\E^{\hat{\sigma}_t^f,\tsigma_{t+1:T}^{f} \tsigma_{t:T}^{l},\pi_t} \big\{ \sum_{n=t}^T R_n^l(Z_n,X^l_n,A^l_n) \big\lvert \pi_t,z_{1:t}, a_{1:t-1}^l, x_{1:t}^l \big\},\label{b_eq:E6}
  }
  where \eqref{b_eq:E2} follows from the fact that probability on $(z_{t+1:T},a^l_{t+1:T},x^l_{2+t:T})$ given $\pi_t,z_{1:t+1}, a_{1:t}^l,x_{1:t+1}^l$ depends on $\pi_{t+1},z_{1:t+1}, a_{1:t}^l,x_{1:t+1}^l$ through ${\sigma_{t+1:T}^{f}, \tsigma_{t+1:T}^{l} }$, \eqref{b_eq:E3} follows Lemma~\ref{l_lemma:1}, \eqref{b_eq:E4} follows from \eqref{b_eq:E1} and the definition of $\hat{\sigma}_t^f$, \eqref{b_eq:E5} follows from Lemma~\ref{l_lemma:2}, \eqref{b_eq:E6} again follows from the fact that probability on $(z_{t+1:T},a^l_{t+1:T},x^l_{2+t:T})$ given $\pi_t,z_{1:t+1}, a_{1:t}^l,x_{1:t+1}^l$ depends on $\pi_{t+1},z_{1:t+1}, a_{1:t}^l,x_{1:t+1}^l$ through ${\sigma_{t+1:T}^{f}, \tsigma_{t+1:T}^{l} }$. However, this leads to a contradiction since $(\tsigma^l,\tsigma^f,z)$ is a GMFE of the game.
}
}
\end{proof}
\bibliographystyle{IEEEtran}

\end{document}